\documentclass[12pt, english]{amsart}
\usepackage[utf8]{inputenc}
\usepackage{xcolor}
\usepackage[margin=1in]{geometry} 

\newcommand{\sk}[1]{{\color{green!50!black}#1}}

\usepackage{amsmath,amssymb,amsthm}
\usepackage{algorithm}      
\usepackage{algpseudocode}
\usepackage{subcaption}
\usepackage[export]{adjustbox}
\usepackage{soul}
\usepackage{graphicx}
\usepackage{caption}
\usepackage{tikz}
\usetikzlibrary{arrows.meta,positioning, decorations.pathreplacing}
\usepackage{tikz-cd}
\usetikzlibrary{decorations.markings}
\usepackage{url}
\usepackage{accents}
\usepackage{mathrsfs}
\usepackage[]{tikz-cd}
\usepackage[all]{xy}
\usepackage{multicol}
\usepackage{mathtools}

\usepackage{color}
\usepackage[symbol]{footmisc}

\usepackage{booktabs}

\theoremstyle{plain}

\newtheorem{theorem}{Theorem}[section]

\newtheorem{corollary}[theorem]{Corollary}
\newtheorem{lemma}[theorem]{Lemma}

\newtheorem{proposition}[theorem]{Proposition}

\newtheorem*{thm*}{Theorem}

\theoremstyle{definition}

\newtheorem{remark}[theorem]{Remark}

\renewcommand{\thefootnote}{\fnsymbol{footnote}}

\newcommand{\R}{\mathbb{R}}

\newcommand{\hrho}{\hat{\rho}}

\DeclareMathOperator{\SO}{SO}

\DeclareMathOperator{\Tr}{Tr}

\DeclareMathOperator{\Orth}{O}

\title{Provable orbit recovery over $\SO(3)$ from the non-uniform second moment}

\author{Tamir Bendory}
\author{Dan Edidin}
\author{Josh Katz\textsuperscript{$\ast$}}
\author{Shay Kreymer}
\author{Nir Sharon}

\date{\today}
\begin{document}

\begin{abstract}
We study the recovery of an unknown three-dimensional band-limited signal from multiple noisy observations that are randomly rotated by latent elements of $\SO(3)$, where the rotations are drawn from an unknown, non-uniform distribution. Because the rotations are unobserved, only the signal orbit under the rotation group can be recovered. We show that the signal orbit and the rotation distribution are jointly identifiable from the first and second moments. This yields an improved high-noise sample complexity that scales quadratically with the noise variance, rather than cubically as in the uniform-rotation case. We further develop a provable, computationally efficient reconstruction algorithm that recovers the 3-D signal by successively solving a sequence of well-conditioned linear systems. The algorithm is validated through extensive numerical experiments. Our results provide a principled and tractable framework for high-noise 3-D orbit recovery, with potential relevance to cryo-electron microscopy and cryo-electron tomography modeling, where molecules are observed in unknown orientations.
\end{abstract}

\maketitle

\renewcommand{\thefootnote}{$\ast$}
\footnotetext{Approved for Public Release; Distribution Unlimited. Public Release Case Number 26-0162. Affiliation with the MITRE Corporation is for identification purposes only and is not intended to convey or imply MITRE’s concurrence with, or support for, the positions, opinions, or viewpoints expressed by the
author. \copyright 2026 The MITRE Corporation. ALL RIGHTS RESERVED.}

\section{Introduction}

We study the problem of recovering a 3‑D signal from its randomly rotated and noisy observations. Specifically, we aim to estimate an unknown signal $x\in V$ from $n$ samples of the form
\begin{equation} \label{eq:samples}
    y_i = g_i\cdot x + \varepsilon_i, \quad i=1,\ldots,n,
\end{equation}
where $g_i$ are random elements drawn from a distribution $\rho$ on the rotation group $\SO(3)$, and $\varepsilon_i$ are i.i.d. Gaussian noise terms with zero mean and variance $\sigma^2$. We assume that $V$ is a finite-dimensional space of band‑limited functions on $\R^3$ and that $\rho$ is a non‑uniform distribution over $\SO(3)$; precise mathematical models are provided in the following sections. Our goal is to recover $x$ from the observations $y_1,\ldots,y_n$, with the rotations $g_1,\ldots,g_n$ being unknown and the noise level $\sigma$ potentially large. Importantly, since both the signal and the rotations are unknown, $x$ is identifiable only up to a global rotation; that is, only its $\SO(3)$-orbit can be recovered.

The model~\eqref{eq:samples} is a special case of a broader class of problems in which the goal is to estimate a signal from noisy samples subject to unknown group actions. This class is commonly known as \emph{orbit recovery} problems or \emph{multi-reference alignment}, e.g.,~\cite{bandeira2023estimation,bendory2017bispectrum,perry2019sample}. When the noise level is low, a standard approach is to first estimate the latent group elements from the samples, then align the observations accordingly, and finally average to reduce noise~\cite{singer2011angular,singer2011three,perry2018message}. In high-noise regimes, however, reliably estimating these group elements (e.g., rotations) becomes difficult, making alignment-based procedures unreliable~\cite{dadon2024detection}. In such settings, one typically turns to methods that estimate the signal directly while marginalizing over the unknown transformations; two popular methods are the expectation-maximization algorithm and the method of moments~\cite{sigworth1998maximum,fan2023likelihood,xu2025misspecified,sharon2020method,shahverdi2024moment}. This work focuses on the latter.

\vspace{5pt} \textbf{Motivation: Electron-microscopy-based technologies for structural biology.}
The primary motivation for studying orbit recovery problems of the form~\eqref{eq:samples} comes from single-particle cryo-electron microscopy (cryo-EM), the leading technology for determining the 3-D structure of biological macromolecules~\cite{nogales2016development,bai2015cryo}. In cryo-EM, one seeks to recover a 3-D (band-limited) function representing the molecular density from many noisy 2-D tomographic projections. Each projection corresponds to the molecule viewed after an unknown, randomly drawn 3-D rotation, so the observations are related through latent group actions~\cite{bendory2020single}.

Model~\eqref{eq:samples} captures the same central statistical difficulty---estimating a signal from multiple noisy, randomly rotated copies---while abstracting away the additional complication introduced by the projection operator. In Section~\ref{sec:cryo}, we make this connection precise by relating~\eqref{eq:samples} to the cryo-EM forward model and outlining how the methods developed here can be potentially extended to incorporate projections.
 
Model~\eqref{eq:samples} is also closely connected to cryo-electron tomography (cryo-ET), an emerging imaging modality for visualizing cellular architecture, including organelles and macromolecular assemblies~\cite{chen2019complete,schaffer2019cryo,turk2020promise}. In cryo-ET, a 3-D specimen (e.g., a cell) is imaged by collecting a tilt series of 2-D projections at known viewing angles. These projections are then  integrated to produce a 3-D tomogram that preserves the \textit{in situ} spatial arrangement of macromolecules within the cell. Researchers often extract smaller 3-D sub-volumes from the tomogram, called subtomograms, centered on macromolecules of interest.
While real tomograms contain a heterogeneous mixture of structures, it is useful to consider an idealized setting in which the volume contains many instances of a single underlying structure $x$, each appearing in an unknown, random orientation. Under this abstraction, the statistical structure aligns naturally with~\eqref{eq:samples}, which models observations as noisy, randomly transformed copies of a common signal~\cite{xu2024bayesian}.

\vspace{5pt} \textbf{The method of moments and sample complexity.}
A remarkable line of work has revealed a close connection between the method of moments and the information-theoretic limits of orbit recovery. In~\cite{abbe2018estimation,abbe2018multireference,bandeira2023estimation}, it was shown that, in the high-noise regime, the minimal number of samples required for accurate signal estimation (i.e., the sample complexity) scales as $\sigma^{2d}$, where $d$ is the smallest moment order that uniquely determines the signal's orbit from the observations.
For instance, in model~\eqref{eq:samples} with $\SO(3)$ transformations drawn uniformly at random, one has $d=3$, and consequently the sample complexity scales as $\sigma^{6}$~\cite{bendory2025orbit}. More generally, the $\sigma^{6}$ scaling also arises in many orbit recovery models under uniform group actions~\cite{bendory2017bispectrum,bandeira2020optimal,bandeira2023estimation,edidin2025orbit}.

\vspace{5pt} \textbf{Contribution.} 
Much of the existing work on orbit recovery problems focuses on identifiability: determining when a moment-matching formulation, which typically involves a system of polynomial equations, admits a unique solution and its implication on sample complexity. In practice, however, the field remains largely driven by heuristic algorithms that lack rigorous theoretical guarantees; prominent examples include~\cite{bendory2022dihedral,scheres2012relion,punjani2017cryosparc,zhong2021cryodrgn}.
This paper makes a significant step forward by introducing a computationally efficient algorithm for signal recovery from the second moment when the distribution of rotations is non-uniform, accompanied by provable theoretical guarantees. 
We show that, in this setting, the unknown signal and the rotation distribution are identifiable from the second moment, which implies a sample complexity scaling as $\sigma^{4}$. We further develop a computationally efficient, provable, and stable recovery algorithm that reconstructs the 3-D function by successively solving a sequence of well-conditioned linear systems. This result complements our previous work~\cite{bendory2025orbit}, which addressed recovery under uniformly distributed rotations via the third moment, and generalized~\cite{drozatz2025provable} from $\SO(2)$ to $\SO(3)$. 

\vspace{5pt} \textbf{Paper organization.} 
The paper is organized as follows. Section~\ref{sec:background} reviews relevant harmonic analysis on compact groups and formulates the moments in representation-theoretic terms. Section~\ref{sec:signal_recovery} presents and proves the main results on signal recovery from the second moment. Section~\ref{sec:numerics} introduces the recovery algorithm---based on successively solving linear systems---and demonstrates its effectiveness numerically, including empirical studies of robustness to noise and the impact of the distribution. Finally, Section~\ref{sec:cryo} outlines potential extensions of the approach to the cryo-EM problem.

\section{Background}
\label{sec:background}

\subsection{Harmonic analysis on compact groups}

A function $\rho\in L^2(G)$ is determined by its {\em Fourier coefficients} for each irreducible representation $V$ of G. If $\varphi \colon G \to  U(V)$ is a unitary representation of $G$, the Fourier coefficient of $V$ is the operator on $V$ defined by
\begin{equation}\label{eq.fourier}
\hrho(V) = \int_{G}\rho(g) \varphi_V(g)^* \;dg,
\end{equation}
where $dg$ is the Haar measure on $G$ with total integral one, $\varphi_V(g)$ is the image of $g$ in $U(V)$, and $\varphi_V(g_i)^*$ denotes the adjoint.
Let $\hat{G}$ be the set of distinct irreducible representations of~$G$. The function $\rho$ is recovered from its Fourier coefficients by the Fourier inversion formula
\begin{equation} \label{eq.fourierinv}
\rho(g) = \sum_{V \in \hat{G}} (\dim V) \Tr(\hrho(V) \varphi_V(g)).
\end{equation}

In the case where $\rho(g)dg$ defines a probability distribution, then the normalization condition $\int_G \rho(g) \;dg =1$ implies that $\hat{\rho}(\mathbf{1})=1$, where $\mathbf{1}$ denotes the trivial representation. The uniform distribution corresponds to the case
where $\hrho(V)$ is the $0$-operator for all non-trivial irreducible representations.
By contrast, for a ``generic" distribution, the Fourier coefficients $\hrho(V)$ are invertible
operators for all irreducibles $V$.

\subsubsection{Harmonic analysis over $\SO(3)$}\label{sec:wigner}

Given our focus on the group $\SO(3)$, we briefly review the spherical harmonics and Wigner $D$-matrices, which provide explicit bases for $L^2(S^2)$ and $L^2(\SO(3))$, respectively.

The \textbf{spherical harmonics} $Y_\ell^m : S^2 \to \mathbb{C}$ for $\ell \geq 0$ and $-\ell \leq m \leq \ell$ form an orthonormal basis for $L^2(S^2)$. Using the Condon--Shortley phase convention, they are defined as
\begin{equation}\label{eq:sph_harm}
Y_\ell^m(\theta, \varphi) = \sqrt{\frac{2\ell+1}{4\pi} \frac{(\ell-m)!}{(\ell+m)!}} P_\ell^m(\cos\theta) e^{im\varphi},
\end{equation}
where $P_\ell^m$ denotes the associated Legendre polynomial. For each $\ell$, the span 
\begin{equation}
    H_\ell = \mathrm{span}\{Y_\ell^m : -\ell \leq m \leq \ell\}
\end{equation} is an irreducible representation of $\SO(3)$ of dimension $2\ell + 1$. The Peter-Weyl decomposition gives
\begin{equation}
L^2(S^2) = \bigoplus_{\ell=0}^{\infty} H_\ell.
\end{equation}
A signal in $L^2(S^2)$ is $L$-band limited if it lies in the subspace $\oplus_{\ell=0}^L H_\ell$. Throughout the paper, we consider band-limited signals and use the notation
$L^2(S^2)_L$ to denote the subspace of 
$L$-band-limited signals in $L^2(S^2)$.

For each $\ell$, the representation $\pi_\ell : \SO(3) \to U(2\ell+1)$ acts on $H_\ell$ by rotation of functions. The matrix elements of $\pi_\ell$ in the spherical harmonic basis are the \textbf{Wigner $D$-functions}:
\begin{equation}
D_{m,m'}^\ell(g) = \langle Y_\ell^m, g \cdot Y_\ell^{m'} \rangle.
\end{equation}

Using the $\text{ZYZ}$ Euler angle parameterization $g = R_z(\alpha) R_y(\beta) R_z(\gamma)$ with $\alpha, \gamma \in [0, 2\pi)$ and $\beta \in [0, \pi]$, the Wigner $D$-functions factor as
\begin{equation}\label{eq:wigner_factor}
D_{m,m'}^\ell(\alpha, \beta, \gamma) = e^{-im\alpha} d_{m,m'}^\ell(\beta) e^{-im'\gamma},
\end{equation}
where $d_{m,m'}^\ell(\beta)$ is the Wigner small-$d$ matrix, which can be expressed in terms of Jacobi polynomials.
The normalized functions $\sqrt{2\ell+1} \, D_{m,m'}^\ell$ form an orthonormal basis for $L^2(\SO(3))$ with respect to the Haar measure. Consequently, any probability density $\rho \in L^2(\SO(3))$ admits the Fourier expansion
\begin{equation}\label{eq:fourier_SO3}
\rho(g) = \sum_{\ell=0}^{\infty} (2\ell+1) \operatorname{Tr}\left(\hat{\rho}(H_\ell) D^\ell(g)\right),
\end{equation}
where the Fourier coefficient $\hat{\rho}(H_\ell) \in \mathbb{C}^{(2\ell+1) \times (2\ell+1)}$ is defined by
\begin{equation}\label{eq:fourier_coeff}
\hat{\rho}(H_\ell)_{m,m'} = \int_{\SO(3)} \rho(g) \overline{D_{m,m'}^\ell(g)} \, dg.
\end{equation}
Throughout the text, we use the notations $\hrho^{\ell}_m[m']$ and $\hat{\rho}(H_\ell)_{m,m'}$ interchangeably to represent the coefficients that represent the distribution.

\subsubsection{Clebsch-Gordan coefficients}
The tensor product of two irreducible representations decomposes as
\begin{equation}\label{eq:CG_decomp}
H_{\ell_1} \otimes H_{\ell_2} \cong \bigoplus_{k=|\ell_1-\ell_2|}^{\ell_1+\ell_2} H_k.
\end{equation}
The \textbf{Clebsch-Gordan (CG) coefficients} $\langle \ell_1 m_1 \, \ell_2 m_2 | \ell_3 m_3 \rangle$ provide the explicit intertwining maps. They satisfy the selection rule $m_1 + m_2 = m_3$ and vanish unless $|\ell_1 - \ell_2| \leq \ell_3 \leq \ell_1 + \ell_2$. These coefficients are central to our second moment formulas, as they govern the projection $H_{\ell_2} \otimes H_{\ell_3} \to H_{\ell_1}$, and they are explicitly given by \begin{align} \label{eq.so3inv}
\langle \ell_1 \, m_1 \, \ell_2 \, m_2 | \ell \, m \rangle = \delta_{m,m_1 + m_2} \sqrt{\frac{(2\ell+1)(\ell+\ell_1-\ell_2)!(\ell-\ell_1+\ell_2)!(\ell_1+\ell_2-\ell)!}{(\ell_1+\ell_2+\ell+1)!}} \nonumber\\\;\times 
\sqrt{(\ell+m)!(\ell-m)!(\ell_1-m_1)!(\ell_1+m_1)!(\ell_2-m_2)!(\ell_2+m_2)!} \nonumber\\ \;\times 
\sum_k \frac{(-1)^k}{k!(\ell_1+\ell_2-\ell-k)!(\ell_1-m_1-k)!(\ell_2+m_2-k)!(\ell-\ell_2+m_1+k)!(\ell-\ell_1-m_2+k)!}.
\end{align}
The following lemma gives the explicit relationship between the projection operators and the Clebsch-Gordan coefficients~\cite{bohm2013quantum}.
\begin{lemma}\label{projection}
For a given $\ell$, let $a^\ell = (a^\ell_{-\ell}, \ldots, a^\ell_\ell) \in H_\ell$ denote a vector 
expanded in the spherical harmonic basis. For $|{\ell_2} - {\ell_3}| \le \ell_1 \le \ell_2 + \ell_3$, the projection $(a^{\ell_2} \otimes a^{\ell_3})_{\ell_1} \in H_{\ell_1}$ has components with respect to the
spherical harmonic basis of $H_{\ell_1}$
\[
\bigl[(a^{\ell_2} \otimes a^{\ell_3})_{\ell_1}\bigr]_{k_1} 
= \sum_{\substack{k_2, k_3 \\ k_2 + k_3 = k_1}} 
\langle \ell_2\, k_2\; \ell_3\, k_3 \mid \ell_1\, k_1 \rangle \, 
a^{\ell_2}_{k_2} \, a^{\ell_3}_{k_3},
\]
where $-\ell_1\leq k_1\leq \ell_1$ and the sum is over $|k_2| \le \ell_2$ and $|k_3| \le \ell_3$.
\end{lemma}

\subsubsection{In-plane uniform distributions}
Next, we introduce and study in-plane-uniform distributions. As will be explained in Section~\ref{sec:signal_recovery}, these distributions play a key role in cryo-EM.

A probability distribution $\rho \in L^2(\SO(3))$ is in-plane uniform if $\rho(gh) = \rho(g)$ for all $h \in \SO(2)$, where $\SO(2)$ is embedded as rotations about the $z$-axis; i.e., matrices of the form
\begin{equation*}
 \begin{pmatrix} \cos \theta & \sin \theta & 0\\ -\sin \theta & \cos \theta & 0 \\ 0 & 0 & 1
\end{pmatrix}, \quad    \theta \in [0,2\pi).
\end{equation*}
 Such distributions correspond to viewing directions that are uniformly distributed in the in-plane angle but may have an arbitrary tilt distribution. 
Note that if $\rho$ is in-plane uniform, then $g \cdot \rho$ is as well,  where $g \cdot \rho \in L^2(\SO(3))$
is the function defined by the formula $(g \cdot \rho)(g') = \rho(g^{-1}g')$ for all $g' \in G$. 
In particular, an in-plane uniform distribution is a function $L^2(\SO(3))^{\SO(2)}$, where
$\SO(2)$ acts on $\rho \in L^2(\SO(3))$ by $(h \rho)(g) = \rho(gh)$. Since this action commutes with the usual
$G$ action on $L^2(\SO(3))$, $L^2(\SO(3))^{\SO(2)}$ is an $\SO(3)$ sub-representation. The following lemma shows that, viewed as a function on $\SO(3)$, elements
of $L^2(\SO(3))^{\SO(2)}$ can be characterized in terms of their Fourier coefficients with respect to the bases of spherical harmonics for each irreducible representation $H_\ell$ of $\SO(3)$.

\begin{lemma}\label{lem:inplane_structure}
With respect to the basis $\{Y^\ell_m\}_{-\ell \leq m \leq \ell}$ of spherical
harmonics for $H_\ell$, the Fourier coefficients of a function $\rho \in L^2(\SO(3))^{\SO(2)}$ satisfy
\begin{equation}\label{eq:inplane_fourier}
\hrho^{\ell}_m[m']= 0 \quad \text{unless } m' = 0.
\end{equation}
That is, only the $m' = 0$ column of each Fourier coefficient matrix is non-zero.
\end{lemma}
\begin{proof}
Since $\rho$ is $H$-invariant, we have that 
\begin{eqnarray*} 
\hrho^{\ell}_m[m']& = &\int_{\SO(3)} \rho(g) \overline{D_{m,m'}^\ell(g)} \, dg\\
& = & \int_{\SO(3)} \rho(gh) \overline{D_{m,m'}^\ell(g)} \, dg\\
& = & \int_{\SO(3)} \rho(g) \overline{D_{m,m'}^\ell(gh)} \, dg,
\end{eqnarray*}
where the second equality follows from a change of variables $g \mapsto gh$.
Now,  
$D_{m,m'}^\ell(g) = \langle Y_\ell^m, g \cdot Y_\ell^{m'} \rangle,$ 
and if $h \in \SO(2)$ corresponds to rotation by an angle $\theta$ about the $z$-axis,
\begin{eqnarray*}
D_{m,m'}^\ell(gh) & = &\langle Y_\ell^m, (g h)\cdot Y_\ell^{m'} \rangle\\
& = & \langle Y_\ell^m, g\cdot ( h \cdot Y_\ell^{m'}) \rangle\\
& = & e^{\iota m'\theta} \langle Y_\ell^m, g\cdot Y_\ell^{m'} \rangle\\
& = & e^{\iota m' \theta} D_{m,m'}^\ell(g),
\end{eqnarray*}
since the action of $\SO(2)$ is diagonalized with respect to the spherical harmonic basis.
Hence, we have that $\hrho^{\ell}_m[m'] = e^{\iota m' \theta} \hrho^{\ell}_m[m']$ for all $\theta \in [0, 2\pi)$,
so if $m' \neq 0$ these coefficients vanish.
\end{proof}

Lemma~\ref{lem:inplane_structure} establishes that the space of $\SO(2)$-invariant
functions is isomorphic to $L^2(S^2)$ 
via the isomorphism $\SO(3)/\SO(2) \cong S^2$.

\subsection{Signal and distribution models}
We model the target signal $x$ as a real-valued, radially discretized band-limited function supported at $R$ spherical shells within the unit ball in $\mathbb{R}^3$. Expanded in spherical harmonics
our function can be represented as:
\begin{equation}
\label{eq:signal}
x(r, \theta, \varphi) = \sum_{\ell=0}^{L} \sum_{m=-\ell}^{\ell} x^{\ell}_m[r] \, Y_\ell^m(\theta, \varphi), 
\end{equation}
where $r=1,\ldots,R$ is a discrete parameter and $\theta, \varphi$ are spherical coordinates. 
As such, our signal space can be identified with the $\SO(3)$-representation 
\begin{equation}
V_L = \bigoplus_{\ell=0}^L H_\ell^R.
\end{equation}
Given a bandlimit of $L$ on the signal and a distribution $\rho \in L^2(\mathrm{SO}(3))$,
the corresponding bandlimited  distribution $\rho_L$ is the function obtained by expanding
the distribution $\rho$ up to band $L$. It can be expressed in the basis of $D$-Wigner matrices as follows:
\begin{equation} \label{eq:distribution}
\rho_L(g) = 
\sum_{\ell=0}^{L}
\;\sum_{m=-\ell}^{\ell}
\;\sum_{m'=-\ell}^{\ell}
\hrho^{\ell}_{m}[m']\, D^{\ell}_{m,m'}(g), \quad g\in \SO(3).   
\end{equation}

\subsection{A Fourier-theoretic perspective of $\SO(3)$ moments}
\label{sec:moments}
Here, we consider the first two population moments
\begin{align} 
M_1(\rho,x) \coloneqq \int_G \rho(g)(gx) dg  \label{eq:population_first_moment}, \\
M_2(\rho,x) \coloneq
\int_G \rho(g)(gx)^{\otimes 2} dg.\label{eq:population_second_moment}
\end{align}

We represent
$x \in V_L=\bigoplus_{\ell=0}^L H_\ell^R$ as an $(L+1)$-tuple of \emph{signal matrices}
$$x = (X_0, \ldots , X_L),$$
where $X_\ell$ is a $(2\ell+1) \times R$ matrix of spherical harmonic coefficients $X_{\ell}=(x_m^{\ell}[r])$
for $m=-\ell,\ldots,\ell$ and $r=1,\ldots,R$.

\textbf{First moment.} The first non-uniform population moment is an element of $V$ which 
is represented by the $(L+1)$-tuple of matrices
\begin{equation} \label{eq.firstmomentdecomp}
M_1(\rho,x) = (\hrho(H_0)^* X_0, \ldots , \hrho(H_\sk{L})^* X_L).
\end{equation}
Hence, if we know the first moment, we know each of the matrices 
$\hrho(H_\ell)^* X_\ell$.
In terms of the Wigner and spherical harmonic bases, the $(m',s)$ entry of $\hrho(H_\ell)^* X_\ell$ is:
\begin{equation}\label{1stmomequation}
(\hrho(H_\ell)^* X_\ell)_{(m',s)} = \sum_{m=-\ell}^{\ell} \overline{\hat{\rho}^\ell_m[m']} \, x^\ell_{m}[s].
\end{equation}
Note that, because $\rho$ is a probability distribution, 
$\int_G \rho(g)\;dg =1$, so $X_0$ is always known from the first moment.

\textbf{Second moment.}
The tensor $V \otimes V$ decomposes
as a sum of tensor products $H_i^R \otimes H_j^R$.
In turn, we can view $H_i^R \otimes H_j^R$ as a sum of $R^2$ summands all equal to $H_i \otimes H_j$. The summand
$H_i \otimes H_j$ decomposes into irreducibles as 
\begin{equation} \label{eq:harmonic_decomp}
(H_i \otimes H_j) = \bigoplus_{\ell = |i-j|}^{i+j} H_{_\ell}.
\end{equation}
Thus, the tensor $X_i \otimes X_j$ decomposes
into components $(X_i \otimes X_j)_{\ell}$, which live in $(H_\ell)^{R^2}$.
Viewed this way, the second moment of $x \in V$, which is an element of $V \otimes V$, decomposes 
as a sum 
\begin{equation} \label{eq.secondmomentdecomp}
M_2(\rho, x) = \sum_{i,j} \left(\sum_{\ell=|i-j|}^{i+j} \hrho(H_\ell)^*(X_i \otimes X_j)_\ell\right),
\end{equation}
where the inner sum reflects the decomposition of $H_i^R \otimes H_j^R$ into its isotypical pieces following
the decomposition~\eqref{eq:harmonic_decomp}. In particular, if we know the second moment we know each of 
the summands $\hrho(H_\ell)^*(X_i \otimes X_j)_\ell$ in the decomposition~\eqref{eq.secondmomentdecomp}. We will denote
by $M_2(x;H_i,H_j,H_k)$ the summand $\hrho(H_k)^*(X_i \otimes X_j)_k$.

The columns of the matrices $X_i$ are the vectors $x^i[s] \in H_i$ for $s =1, \ldots ,R$.
Thus, the tensor $X_i \otimes X_j$ can be viewed as a collection of $R^2$ tensors $x^i[s_1] \otimes x^j[s_2]
\in H_i \otimes H_j$
for $1 \leq s_1, s_2 \leq R$. 
Applying Lemma~\ref{projection} to the second moment~\eqref{eq.secondmomentdecomp}, we obtain an explicit formula in terms of Clebsch-Gordan coefficients for the action of the operator
$\hrho(H_\ell)^*$ on $(x^i[s_1] \otimes x^j[s_2])_\ell$.
To make this precise, we change notation slightly and fix bands $\ell_1, \ell_2, \ell_3$ satisfying the triangle inequality and shell indices $s_2, s_3 \in \{1, \ldots, R\}$. Then,
for $-\ell_1 \leq m \leq \ell_1$, the $m$-th entry in the vector
$\hat{\rho}(H_{\ell_1})^* (x^{\ell_2}[s_2] \otimes x^{\ell_3}[s_3])_{\ell_1}$ is
\begin{align}\label{2ndmomequation}
M_2(\rho, x; H_{\ell_1}[m], H_{\ell_2}[s_2], H_{\ell_3}[s_3]) \coloneq \sum_{\substack{k_2 + k_3 = k_1 \\ |k_i| \le \ell_i}}
\langle \ell_2\, k_2\; \ell_3\, k_3 \mid \ell_1\, k_1 \rangle \,
\overline{\hat{\rho}^{\ell_1}_{k_1}[m]} \, x^{\ell_2}_{k_2}[s_2] \, x^{\ell_3}_{k_3}[s_3].
\end{align}

\section{Main results} 
\label{sec:signal_recovery}

 We begin this section by explicitly introducing the signal and distribution models. These models will also serve us in the numerical experiments in Section~\ref{sec:numerics}. Then, we state the two main theorems on signal recovery from the second moment for a generic distribution on $\SO(3)$ (Theorem~\ref{thm:main}) and one for a generic in-plane uniform distribution on $\SO(3)$ (Theorem~\ref{thm:inplane}). The results are then proved in Section~\ref{sec:thm_main} and Section~\ref{sec:thm_inplane}.

\subsection{Main results}
We are now ready to present the main theoretical results of this paper. 

\begin{theorem}\label{thm:main}
Let $\rho \in L^2(\SO(3))$ be a distribution whose Fourier coefficients $\hrho(H_\ell)$ are invertible
for $0 \leq \ell \leq L$. 
Then, if $R \geq 3$, and $x \in V_L$ is a generic signal, then the pair $(\rho_L,x)$ is determined, up to a global rotation, by the first and second moments $M_1(\rho,x)$ and $M_2(\rho,x)$.
\end{theorem}

By generic signal we mean that
there is a non-empty Zariski open subset of $V_L$ such that the conclusion of the theorem holds. Note that this set may depend on the distribution $\rho$.

While Theorem~\ref{thm:main} assumes that the Fourier coefficients of the distribution are invertible
(at least up to band $L$), this assumption may not be experimentally justified. In particular, in cryo‑EM the distribution of in‑plane rotations is uniform, since each in‑plane rotation is equally likely to occur~\cite{bendory2020single}. The following theorem shows that in this important scenario---where the in‑plane rotation distribution is uniform but otherwise generic---the theory continues to hold, with only a modest increase in the required number of shells.

\begin{theorem}\label{thm:inplane} Let $\rho$ be an in-plane uniform distribution 
whose non-vanishing Fourier coefficients $\{\hrho_m^\ell[0]\}$ with $0 \leq \ell \leq L$ 
are generic. If $R \geq 4$, then, for a generic signal
$x \in V_L$,
the pair $(\rho_L,x)$ is determined, up to a global rotation, by the first and second moments $M_1(\rho,x)$ and $M_2(\rho,x)$.
\end{theorem}

By generic Fourier coefficients and generic signal we mean the following. If we let $\mathcal{U}_L$ denote the space of all possible 
$\{\hrho_m^\ell[0]\}$ with $0 \leq \ell \leq L$, then there is a Zariski open set in $\mathcal{U}_L$ such that the following holds:
For every in-plane uniform distribution $\rho$ whose non-vanishing Fourier coefficients lie in this Zariski open set,  there is a corresponding Zariski open set in the signal space $V_L$ such that for $x$ in this open set the conclusion of the theorem holds.

\begin{corollary}
Consider the class of distributions and signals defined in Theorems~\ref{thm:main} and~\ref{thm:inplane}.
In the high-noise regime, where $n,\sigma \to \infty$, the sample complexity required to estimate a signal from observations of~\eqref{eq:samples} scales as $n = \omega(\sigma^4)$.
\end{corollary}

 \begin{remark}[The optimality of the results] It is unclear whether the results of Theorem~\ref{thm:main} and Theorem~\ref{thm:inplane} are optimal. In particular, the bounds of $R=3$ and $R=4$  may be an artifact of the proof technique of the specific algorithm we analyze.
\end{remark}

\subsection{Proof of Theorem~\ref{thm:main}} 
\label{sec:thm_main}
Our proof of Theorem~\ref{thm:main} is inductive and constructive. For each $\ell$, we first use the second moment to solve for the generalized Fourier coefficient $\hrho(H_\ell)$, assuming that Fourier coefficients $\hrho(H_{j})$ and signal coefficients $X_{j}$ are known for $j \leq l$. Having done this, we can easily determine the signal coefficient matrix $X_\ell$ from the first moment because the first moment yields $\hrho(H_\ell)^* X_\ell$. Hence, if $\hrho(H_\ell)$ is invertible we can determine $X_\ell$ as well. This alternating procedure is illustrated in Figure~\ref{fig:band-ladder}.

\par\noindent
\begin{minipage}{\linewidth}
\centering
\begin{tikzpicture}[
  font=\small,
  sigbox/.style={draw, rounded corners=4pt, align=center, inner sep=6pt, 
                 minimum width=2.8cm, minimum height=1cm,
                 fill=blue!8, draw=blue!50!black, line width=0.6pt},
  distbox/.style={draw, rounded corners=4pt, align=center, inner sep=6pt, 
                  minimum width=2.8cm, minimum height=1cm,
                  fill=orange!10, draw=orange!60!black, line width=0.6pt},
  knownbox/.style={draw, rounded corners=4pt, align=center, inner sep=6pt, 
                   minimum width=2.8cm, minimum height=1cm,
                   fill=gray!15, draw=gray!60, line width=0.6pt},
  mainflow/.style={-{Latex[length=2.5mm]}, thick, black!80},
  depflow/.style={-{Latex[length=2mm]}, densely dashed, gray!60, thin},
  steplabel/.style={font=\scriptsize\bfseries, fill=white, inner sep=2pt, circle},
]

\node[font=\normalsize\bfseries, blue!50!black] at (-2.2, 1) {Signal};
\node[font=\normalsize\bfseries, orange!60!black] at (2.2, 1) {Distribution};

\node[knownbox] (x0) at (-2.2, 0) {$X_0$};
\node[knownbox] (rho0) at (2.2, 0) {$\hat{\rho}(H_0)$};
\node[knownbox] (x1) at (-2.2, -1.8) {$X_1$};

\def\rowsep{1.8}

\node[distbox] (rho1) at (2.2, -1.8) {$\hat{\rho}(H_1)$};

\node[distbox] (rho2) at (2.2, {-1.8-1*\rowsep}) {$\hat{\rho}(H_2)$};
\node[sigbox]  (x2)   at (-2.2, {-1.8-1*\rowsep}) {$X_2$};

\node[distbox] (rho3) at (2.2, {-1.8-2*\rowsep}) {$\hat{\rho}(H_3)$};
\node[sigbox]  (x3)   at (-2.2, {-1.8-2*\rowsep}) {$X_3$};

\node[distbox] (rho4) at (2.2, {-1.8-3*\rowsep}) {$\hat{\rho}(H_4)$};
\node[sigbox]  (x4)   at (-2.2, {-1.8-3*\rowsep}) {$X_4$};

\node[font=\large] at (-2.2, {-1.8-4*\rowsep+0.5}) {$\vdots$};
\node[font=\large] at (2.2, {-1.8-4*\rowsep+0.5}) {$\vdots$};


\draw[mainflow] (x0.east) -- node[steplabel, above] {\footnotesize 1} (rho1.north west);
\draw[mainflow] (x1.east) -- (rho1.west);

\draw[mainflow] (x1.south east) -- node[steplabel, above, pos=0.4] {\footnotesize 2} (rho2.north west);
\draw[depflow] (x0.south east) to[out=-20, in=150] (rho2.north);

\draw[mainflow] (rho2.west) -- node[steplabel, above] {\footnotesize 3} (x2.east);

\draw[mainflow] (x2.south east) -- node[steplabel, above, pos=0.4] {\footnotesize 4} (rho3.north west);
\draw[depflow] (x1.south east) to[out=-30, in=160] (rho3.west);

\draw[mainflow] (rho3.west) -- node[steplabel, above] {\footnotesize 5} (x3.east);

\draw[mainflow] (x3.south east) -- node[steplabel, above, pos=0.4] {\footnotesize 6} (rho4.north west);
\draw[depflow] (x2.south east) to[out=-30, in=160] (rho4.west);

\draw[mainflow] (rho4.west) -- node[steplabel, above] {\footnotesize 7} (x4.east);

\draw[decorate, decoration={brace, amplitude=5pt, mirror}, gray!70, line width=0.5pt] 
  (-3.9, 0.4) -- (-3.9, -2.2) node[midway, left=6pt, font=\scriptsize, text=gray!70, align=center] {known\\(base case)};

\draw[decorate, decoration={brace, amplitude=5pt}, gray!70, line width=0.5pt] 
  (3.9, 0.4) -- (3.9, -0.4) node[midway, right=6pt, font=\scriptsize, text=gray!70, align=center] {known};

\end{tikzpicture}

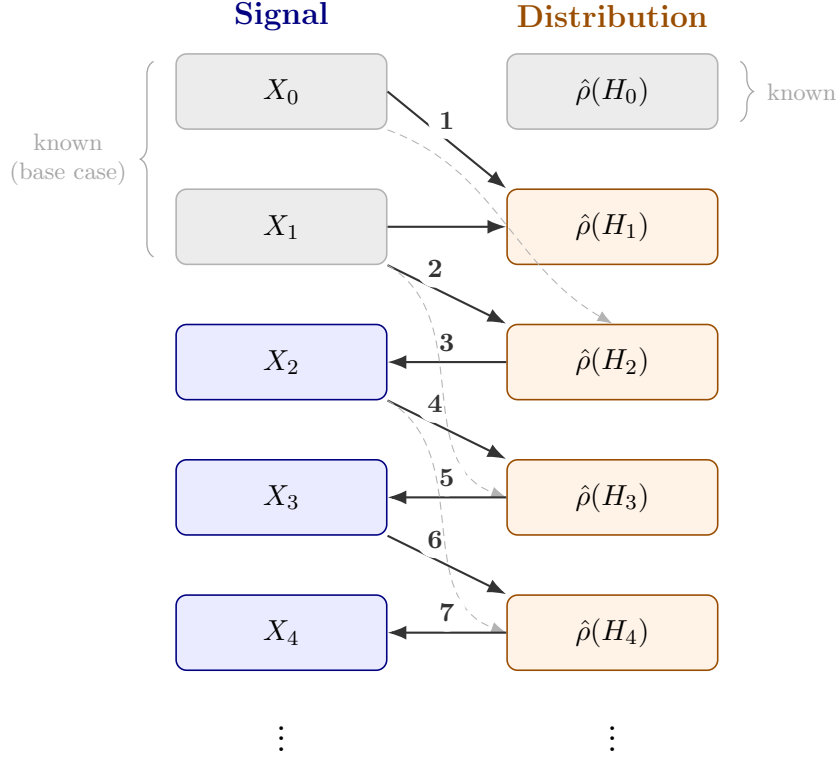
\captionof{figure}{A schematic illustration of the frequency-marching recovery. Signal coefficients (blue, left) and distribution Fourier coefficients (orange, right) are recovered at progressively higher frequency bands. At each band $\ell \geq 2$, the distribution $\hat{\rho}(H_\ell)$ is first solved from the second moment using previously recovered signal coefficients (diagonal and dashed arrows), and then $X_\ell$ is recovered from the first moment by inverting $\hat{\rho}(H_\ell)$ (horizontal arrows). The numbers show the order of computation.}
\label{fig:band-ladder}
\end{minipage}\\

 \paragraph{\bf{Base case.}}
 We begin by recalling that for any compact group, the $G$-invariant part 
of the $d$-th moment with respect to an arbitrary probability distribution is the uniform $d$-th moment. In particular, the uniform second moment is known from the non-uniform second moment, so by~\cite{bendory2024sample} we know the matrices $X^*_\ell X_\ell$ for each $\ell$. Moreover, we know $X_0$ since we know $\hrho(\mathbf{1})=1$ for any probability distribution and we know $\hrho(\mathbf{1})X_0$. Let $X'_1$
 be any square root of the Gram matrix $X^*_1 X_1$. Once we fix $X'_1$,  we can uniquely determine the Fourier coefficient $\hrho(H_1)$ from the first moment $\hrho(H_1)^* X'_1$ since we assume
 that we have at least three shells ($R\geq 3$) and the latter implies that the rank of $X'_1$  is three for a generic signal. 
Hence, we have recovered the first bands of both the distribution and signal up to the same $\mathrm{O}(3)$ action,
which is equivalent to recovering their joint $\SO(3)$ orbit up to a global sign.
We can eliminate the sign ambiguity by comparing
$(\hrho(H_1)')^*(X'_1 \otimes X'_1)_1$ with the component $\hrho(H_1)^*(X_1 \otimes X_1)_1$ of the second moment. 

 \paragraph{\bf{Induction step.}} Now, suppose that we have determined the following Fourier coefficients,
 $\hrho(\mathbf{1}), \ldots , \hrho(H_{\ell-1})$, and thus the corresponding signal matrices $X_1,\ldots, X_{\ell-1}$.
 We will show that we can solve a linear system to determine the Fourier coefficient $\hrho(H_{\ell})$. 
 \begin{lemma}\label{lem:inductive}
Suppose the orbit of $(\rho, x) \in L^2(\mathrm{SO}(3))_{\ell-1} \times (L^2(S^2)_{\ell-1})^R$ is fixed with $\ell > 1$. If $-\ell \leq s \leq \ell$, the distribution coefficients 
${\rho}^\ell_k[s]$ for $-\ell \leq k \leq \ell$
are determined by the invariants $M_2(\rho, x; H_\ell[s_1], H_i[s_2], H_{\ell-i}[s_3])$ for $0 \leq i \leq \lfloor \ell/2 \rfloor$, provided that $R \geq 3$.
\end{lemma}

\begin{proof}
For fixed $s_1$, the second moment equation~\eqref{2ndmomequation} 
\[
M_2(\rho, x; H_\ell[s_1], H_i[s_2], H_{\ell-i}[s_3]) = \sum_{\substack{k_2 + k_3 = k_1 \\ |k_1| \le \ell, |k_2| \le i, |k_3| \le \ell-i}}
\langle i\, k_2\; (\ell-i)\, k_3 \mid \ell\, k_1 \rangle \,
\overline{\hat{\rho}^{\ell}_{k_1}[s_1]} \, x^{i}_{k_2}[s_2] \, x^{\ell-i}_{k_3}[s_3],
\]
gives a linear equation in the $2\ell +1$ unknowns $\overline{\hrho}^\ell_k[s_1]$ with $-\ell \leq k \leq \ell$.
The coefficients $C(x^{i}[s_2], x^{\ell-i}[s_3])_k = \sum_{\substack{k_1 - k_2 = k}} \langle i\, k_1\; (\ell-i)\, k_2 \mid \ell\, k \rangle \, x^{i}[s_2]_{k_1} x^{\ell-i}[s_3]_{k_2}$ are bilinear forms in the variables 
$x^{i,\ell-i}_{s,t}[s_2,s_3] \coloneq x^{i}_s[s_2] x^{\ell -i}_t[s_3]$.
Taking indices $(s_2, s_3) \in S_{i,\ell-i}[R]$, where
\[
S_{i,\ell-i}[R] = \begin{cases}
\{(1,1), \ldots, (1,R), (2,1), \ldots, (R,1)\} & \text{if } i \neq \ell-i, \\
\{(1,1), \ldots, (1,R)\} & \text{if } i = \ell-i,
\end{cases}
\]
and varying $i$ from $0$ to $\lfloor \ell/2 \rfloor$, we obtain
$(\lfloor \ell/2 \rfloor + 1) \cdot (2R - 1)$
equations.
For $R \geq 3$, this gives at least $(\lfloor \ell/2 \rfloor + 1) \cdot 5 \geq 2\ell + 1$ when $\ell \geq 1$. 

For two distinct pairs $(s^{\prime}_2, s^{\prime}_3), (s^{\prime\prime}_2, s^{\prime\prime}_3) \in S_{i,\ell-i}[R]$, there are no algebraic relations between the corresponding sets of monomials, so we may treat them as distinct sets of variables.
Using the shorthand $C(\{x^{i,\ell-i}_{s,t}[s_2,s_3]\})_k$ for the coefficient $C(x^{i}[s_2], x^{\ell-i}[s_3])_k$, the coefficient matrix is
\begin{equation} \label{eq:coeffmatrix}
M = \begin{pmatrix}
C(\{x^{0,\ell}_{s,t}[1,1]\})_{-\ell} & \cdots & C(\{x^{0,\ell}_{s,t}[1,1]\})_{\ell} \\
\vdots & & \vdots \\
C(\{x^{0,\ell}_{s,t}[1,R]\})_{-\ell} & \cdots & C(\{x^{0,\ell}_{s,t}[1,R]\})_{\ell} \\
C(\{x^{0,\ell}_{s,t}[2,1]\})_{-\ell} & \cdots & C(\{x^{0,\ell}_{s,t}[2,1]\})_{\ell} \\
\vdots & & \vdots \\
C(\{x^{0,\ell}_{s,t}[R,1]\})_{-\ell} & \cdots & C(\{x^{0,\ell}_{s,t}[R,1]\})_{\ell}\\
C(\{x^{1,\ell-1}_{s,t}[1,1]\})_{-\ell} & \cdots & C(\{x^{1,\ell-1}_{s,t}[1,1]\})_{\ell} \\
\vdots & & \vdots \\
C(\{x^{1,\ell-1}_{s,t}[1,R]\})_{-\ell} & \cdots & C(\{x^{1,\ell-1}_{s,t}[1,R]\})_{\ell} \\
C(\{x^{1,\ell-1}_{s,t}[2,1]\})_{-\ell} & \cdots & C(\{x^{1,\ell-1}_{s,t}[2,1]\})_{\ell} \\
\vdots & & \vdots \\
C(\{x^{1,\ell-1}_{s,t}[R,1]\})_{-\ell} & \cdots & C(\{x^{1,\ell-1}_{s,t}[R,1]\})_{\ell}\\
\vdots & & \vdots \\
C(\{x^{\lfloor \ell/2 \rfloor, \lceil \ell/2 \rceil}_{s,t}[1,1]\})_{-\ell} & \cdots & C(\{x^{\lfloor \ell/2 \rfloor, \lceil \ell/2 \rceil}_{s,t}[1,1]\})_{\ell} \\
\vdots & & \vdots \\
C(\{x^{\lfloor \ell/2 \rfloor , \lceil \ell/2 \rceil}_{s,t}[1,R]\})_{-\ell} & \cdots &  C(\{x^{\lfloor \ell/2 \rfloor , \lceil \ell/2 \rceil}_{s,t}[1,R]\})_{\ell}\\
C(\{x^{\lfloor \ell/2 \rfloor, \lceil \ell/2 \rceil}_{s,t}[2,1]\})_{-\ell} & \cdots & C(\{x^{\lfloor \ell/2 \rfloor, \lceil \ell/2 \rceil}_{s,t}[2,1]\})_{\ell} \\
\vdots & & \vdots \\
C(\{x^{\lfloor \ell/2 \rfloor, \lceil \ell/2 \rceil}_{s,t}[R,1]\})_{-\ell} & \cdots & C(\{x^{\lfloor \ell/2 \rfloor, \lceil \ell/2 \rceil}_{s,t}[R,1]\})_{\ell} 
\end{pmatrix}.
\end{equation}
Since the forms $C(\{x^{i,\ell-i}_{s,t}[s_2,s_3]\})_k$ for $k = -\ell, \ldots, \ell$ are linearly independent for each fixed $(s_1, s_2)$, the matrix $M$ has full rank over the field of rational functions. Thus, for generic $(\rho, x)$, the matrix has rank $2\ell+1$.
In particular, the system has a unique solution for the entries of the Fourier coefficient matrix
$\hrho(H_\ell)$. This completes the proof of Lemma~\ref{lem:inductive} and with it Theorem~\ref{thm:main}.
\end{proof}

\subsection{Proof of Theorem~\ref{thm:inplane}}
\label{sec:thm_inplane}

The proof of Theorem~\ref{thm:inplane} has a similar structure to the proof of Theorem~\ref{thm:main}. However, because the Fourier coefficients of an in-plane uniform distribution all have rank one, we cannot use the first moment to solve for the unknown matrices $X_\ell$ if we already know the Fourier coefficient
$\hrho(H_\ell)$. Because $\hat{\rho}^{\ell}_m[m'] = 0$ if $m' \neq 0$, the only non-zero terms in the second moment have the form 
\[M_2(\rho, x;H_i[0], H_{j}[s_2], H_k[s_3]),\]
for varying values of $i,j,k \in [0,L]$ and $s_2, s_3 \in [1,R]$. \\

\paragraph{{\bf Base case.}} The base case is essentially the same as in the proof Theorem~\ref{thm:main}. Because we have at least three shells, we may assume that the matrix $X_1$ has rank at least three and can therefore solve for the three non-zero entries in $\hrho(H_1)$ from the first moment $\hrho(H_1) X_1$.\\

\paragraph{{\bf Induction step.}}
Let $\ell > 1$ and suppose all terms of the distribution and signal have been determined up to band $\ell-1$, i.e., $\{\hrho^i_k[0]\}$ and $\{x^i_k[s]\}$ are fixed for $1 \leq i \leq \ell-1$, $-i\leq k \leq i$
and $1 \leq s \leq R$.\\

\noindent{\textbf{Step 1: The case $\ell = 2$.}} 
Let $\hat{\rho}(H_1) \in \mathbb{R}^{3 \times 3}$ and $X_1 \in \mathbb{R}^{3 \times R}$ be fixed. We first solve for the distribution term, $\hat{\rho}(H_2)$. We use the equations $\{M_2(H_2[0], H_1[s_2], H_1[s_3])\}$ with indices taking values in $1 \leq s_2 \leq R$, $s_3 = 1$ and $s_2 = 1$, $2 \leq s_3 \leq R$. This yields $2R - 1$ equations, which is at least $\dim(H_2) = 5$ when $R \geq 3$.

Now we solve for the signal component $X_2$. For fixed $1 \leq s_3 \leq R$, the $R$ equations
$\{M_2(\rho,x;H_{1}[0], H_1[s_2], H_2[s_3])\}$ with $1 \leq s_2 \leq R$
and the equation 
$\{M_2(\rho,x;H_{2}[0], H_1[1], H_2[s_3])\}$
involve algebraically independent coefficients. Therefore, we
obtain $R+1\geq \dim(H_2)=5$ (when $R\geq 4$) independent equations, which allows us to 
solve for $x^2[s_3]$.

\noindent{\textbf{Step 2: The case $\ell > 2$.}}
Fix $1 \leq s_3 \leq R$. We solve for the $2\ell+1$ unknowns $\{x^\ell_k[s_3]\}_{k=-\ell}^{\ell}$ using invariants of the form $M_2(H_i[0], H_{\ell-i}[s_2], H_\ell[s_3])$ for $1 \leq i \leq \ell-1$ and $1 \leq s_2 \leq R$.

For fixed $i$ and  $s_2$, the second moment equation~\eqref{2ndmomequation}
\[
M_2(\rho, x; H_i[0], H_{\ell-i}[s_2], H_\ell[s_3]) = 
\sum_{\substack{k_2 + k_3 = k_1 \\ |k_1| \le i, |k_2| \le (\ell-i), |k_3| \le i}}
\langle (\ell-i) \, k_2\; \ell\, k_3 \mid i\, k_1 \rangle \,
\overline{\hat{\rho}^{i}_{k_1}[0]} \, x^{\ell-i}_{k_2}[s_2] \, x^{\ell}_{k_3}[s_3].\]
gives a linear equation in the unknowns $x^\ell_k[s_3]$, with coefficients 
\[C(\overline{\hrho^i[0]}, x^{\ell-i}[s_2])_k = 
\sum_{\substack{k_1 - k_2 = k}} \langle (\ell-i)\, k_1\; \ell\, k \mid i\, k_1 \rangle \, 
\overline{\hrho^i_{k_1}[0]} x^{\ell-i}_{k_2}[s_2]
\] that are bilinear forms in the variables $y^{i,\ell-i}_{s,t}[s_2] \coloneq \overline{\hrho^i_{s}[0]} \, x^{\ell-i}_t[s_2]$.

Let 
\[
S[R] = \{(i, s_2) : 1 \leq i \leq \ell-1, \, 1 \leq s_2 \leq R\} \subset [1, \ell-1] \times [1, R].
\]
Taking $(i, s_2) \in S[R]$ yields $|S[R]| = (\ell-1) \cdot R$ linear equations. For distinct pairs $(i, s_2), (i', s_2') \in S[R]$, there are no algebraic relations between the corresponding sets of monomials $\{y^{i,\ell-i}_{s,t}[s_2]\}_{s,t}$ and $\{y^{i',\ell-i'}_{s,t}[s_2']\}_{s,t}$: if $i \neq i'$ the variables lie in different bands, and if $i = i'$ but $s_2 \neq s_2'$ the variables lie in different shells. Thus, we may treat them as distinct sets of variables.

Using the shorthand $C(\{y^{i,\ell-i}_{s,t}[s_2]\})_k$ for the coefficient 
$C(\overline{\hrho^i[0]}, x^{\ell-i}[s_2])_k$
the coefficient matrix is
\begin{equation}
\label{eq:coeff_matrix}
M = \begin{pmatrix}
C(\{y^{1,\ell-1}_{s,t}[1]\})_{-\ell} & \cdots & C(\{y^{1,\ell-1}_{s,t}[1]\})_{\ell} \\
\vdots & & \vdots \\
C(\{y^{1,\ell-1}_{s,t}[R]\})_{-\ell} & \cdots & C(\{y^{1,L-1}_{s,t}[R]\})_{\ell} \\
C(\{y^{2,L-2}_{s,t}[1]\})_{-\ell} & \cdots & C(\{y^{2,L-2}_{s,t}[1]\})_{\ell} \\
\vdots & & \vdots \\
C(\{y^{L-1,1}_{s,t}[R]\})_{-\ell} & \cdots & C(\{y^{L-1,1}_{s,t}[R]\})_{\ell}
\end{pmatrix}.
\end{equation}
Since the forms $C(\{y^{i,\ell-i}_{s,t}[s_2]\})_k$ for $k = -\ell, \ldots, \ell$ are linearly independent for each fixed $(i, s_2)$, the matrix $M$ has full rank over the field of rational functions. Thus, for generic $(\rho, x)$, the matrix has rank $\min\left(\ell-1) \cdot R, 2\ell+1\right)$.
The system has a unique solution when $(\ell-1) \cdot R \geq 2\ell + 1$. For $R = 4$, this gives $4(\ell-1) = 4\ell - 4 \geq 2\ell + 1$, which holds when $\ell \geq 3$.

\medskip
\noindent\textbf{Recovery of the distribution band $\ell$.}
By Lemma~\ref{lem:inductive},  if the signal coefficients are generic,  then the distribution
coefficients $\hrho^{\ell}_m[0]$ are determined from the second moment invariants, 
$M_2(\rho, x; H_\ell[0], H_i[s_2], H_{\ell-i}[s_3])$ for $0 \leq i \leq \lfloor \ell/2 \rfloor$, provided that $R \geq 3$. Hence, for a generic distribution and signal $(\rho, x)$, we can recover $x$ up to the action of $\mathrm{SO}(3)$ from the second moment.

\section{Algorithm and Numerical experiments}\label{sec:numerics}
In this section, we first explicitly introduce the algorithm derived from the proofs in Section~\ref{sec:signal_recovery}. Then, we empirically  validate the theoretical results and  demonstrate that our frequency marching algorithm, based on the second moment, successfully recovers the signal when the  distribution is non-uniform. Furthermore, we compare its performance against the third-moment method proposed in~\cite{bendory2025orbit}, analyzing the robustness of both estimators with respect to noise, sample complexity, and the degree of nonuniformity in the data. 

\subsection{Signal model}
\label{sec:setting}
In the numerical part, we model the signal as
\begin{equation}
\label{eq:3-d-fourier-bessel}
x(r, \theta, \varphi) = \sum_{\ell=0}^{L} \sum_{m=-\ell}^{\ell} \sum_{s=1}^{R} x_m^{\ell}[s] \, Y_\ell^m(\theta, \varphi) \, j_{\ell, s}(r).
\end{equation}
This model is similar to~\eqref{eq:signal}, but the function is not discretized in the radial direction. Specifically, the radial direction is described by the normalized spherical Bessel functions $j_{\ell, s}(r)$ defined as $j_{\ell, s}(r) = \frac{\sqrt{2}}{\lvert j_{\ell + 1}(u_{\ell, s}) \rvert} \, j_\ell(u_{\ell, s} r)$, where $j_\ell$ is the spherical Bessel function of the first kind of order $\ell$, and $u_{\ell, s}$ denotes its $s$-th positive zero. The number of radial components is limited to a fixed number of shells~$R$,
and such a function can also be represented as an element in the finite-dimensional $\SO(3)$-representation $V_L = \bigoplus_{\ell=0}^L H_\ell^R$. This representation is chosen for its superior numerical stability compared to the shell model~\cite{kileel2025fast}, and the theoretical results established above follow directly in this setting.
 The representation of the distribution is as in~\eqref{eq:distribution}. 

\subsection{Algorithm}
\label{sec:algorithm} 
The algorithm for signal recovery from the second moment is induced by our proofs in Section~\ref{sec:signal_recovery}, with slight changes that are explained below.
It is based on a frequency marching mechanism, where the $\ell$-th frequency of the signal and the distribution is recovered linearly based on the knowledge of lower frequencies. A schematic illustration of this alternating procedure is given in Figure~\ref{fig:band-ladder}.
We now formulate the algorithm explicitly. 

The algorithm begins by estimating the population moments from the empirical moments. While the first moment is unbiased and can therefore be estimated directly, the second moment includes a bias term. However, this bias affects only a small portion of the second moment, specifically, the component that is invariant under $\SO(3)$.

\begin{lemma}\label{prop:unbiased}
Let $\{y_i=g_ix+\varepsilon_i\}_{i=1}^n\subset (L^2(S^2)_L)^R$ be a collection of $n$ noisy, randomly rotated copies of the signal $x$. Let $y_i^{\ell}[s]\in H_{\ell}$ denote the components of the sample in the s-th radial shell of the $\ell$-th band. For $\ell_1 \ge 1$, we have 
\[
\frac{1}{n} \sum_{i=1}^n (y_i^{\ell_2}[s_2] \otimes y_i^{\ell_3}[s_3])_{\ell_1} 
\;\xrightarrow{n \to \infty}\; 
\hat{\rho}(H_{\ell_1})^* (x^{\ell_2}[s_2] \otimes x^{\ell_3}[s_3])_{\ell_1}.
\]
For $\ell_1 = 0$, we obtain
\[
\frac{1}{n} \sum_{i=1}^n (y_i^{\ell}[s] \otimes y_i^{\ell}[s])_0 
\;\xrightarrow{n \to \infty}\; 
x^\ell[s]x^\ell[s]^* + \sigma^2.
\]
\end{lemma}

\begin{proof}
The noise $\varepsilon_i$ contributes $\sigma^2 I_{2\ell+1}$ to diagonal blocks $H_\ell \otimes H_\ell$ of the second moment. Under the decomposition $H_\ell \otimes H_\ell = \bigoplus_{k=0}^{2\ell} H_k$, the identity $I_{2\ell+1}$ lies entirely in the trivial summand $H_0$. Projection onto $H_k$ for $k \ge 1$ therefore annihilates the bias.
\end{proof}

Hence, the empirical second moment $\frac{1}{n}\sum_{i=1}^n y_i \otimes y_i$ is biased by $\sigma^2 I$ due to the noise, and we therefore debias the estimator as follows:
\begin{equation}
    \label{eq:debias}
    \hat{M}_2 = \frac{1}{n}\sum_{i=1}^n y_i \otimes y_i - \sigma^2 I.
\end{equation}
The empirical first moment is simply given by
\begin{equation}
    \hat{M}_1 = \frac{1}{n}\sum_{i=1}^n y_i.
\end{equation}

Next, we aim to recover the signal and distribution coefficients. First, the signal coefficients $x_0^0[s]$, $1 \le s \le R$, are recovered directly from the $\ell=0$ component of the first moment~\eqref{eq.firstmomentdecomp}, since $\hat{\rho}(H_0) = 1$:
\begin{equation}
    \label{eq:x_0}
    \hat{x}_0^0[s] = {(\hat{M}_1)_0}_{(0, s)}.
\end{equation}
The coefficients~$x_m^1[s], 1 \le s \le R, m \in \{-1, 0, 1\}$, are determined up to $\Orth(3)$ from the Gram matrix whose $(s,s')$ entry is 
\begin{equation}
    \label{eq:gram_1}
    \sum_{m=-1}^{1} x_m^1[s]\, \left({x_m^1}[s']\right)^*,
\end{equation}
which is obtained from the uniform second moment; see the base case of the proof of Theorem~\ref{thm:main}. The coefficients are recovered up to a global rotation and reflection by computing a matrix square root, e.g., via eigendecomposition. We can eliminate the reflection by utilizing one additional equation from the second moment. Then, from the second moment equations~\eqref{2ndmomequation}, we proceed:
\begin{enumerate}
    \item When the signal coefficients $x$ are fixed, \eqref{2ndmomequation} is linear in the distribution coefficients~$\rho$.
    \item When the distribution coefficients $\rho$ are fixed,~\eqref{2ndmomequation} is linear in the signal coefficients of the highest band (when solving progressively).
\end{enumerate}
We note that although the proof of Theorem~\ref{thm:main} uses the first moment to recover signal coefficients, in practice we solve for both signal and distribution coefficients directly from the second moment equations~\eqref{2ndmomequation}, which subsumes the first moment information and yields better numerical performance.

\textbf{Distribution recovery.} Assuming signal coefficients for bands $0, \dots, \ell-1$ are known, we recover the distribution Fourier coefficients at band $\ell$ by solving $(2\ell+1)$ independent linear systems, one for each column $s_1 \in \{-\ell, \dots, \ell\}$ of the distribution Fourier coefficient matrix $\hat{\rho}(H_\ell)$. For each column index $s_1$, equation~\eqref{2ndmomequation} yields a linear system:
\begin{equation}
\label{eq:linear_rho}
A_{\rho}^{(s_1)}(x_{0 \dots \ell-1}) \, \overline{\hat{\rho}_{\ell}[s_1]} = b_{\rho}^{(s_1)}(\hat{M}_2),
\end{equation}
where $\hat{\rho}_{\ell}[s_1] \in \mathbb{C}^{2\ell+1}$ is the $s_1$-th column of $\hat{\rho}(H_\ell)$. The matrix $A_\rho^{(s_1)}$ has $(2\ell+1)$ columns indexed by $m_1 \in \{-\ell, \ldots, \ell\}$, and rows indexed by quadruples $(\ell_2, \ell_3, s_2, s_3)$ where $|\ell - \ell_2| \le \ell_3 \le \min(\ell + \ell_2, \ell)$ and $s_2, s_3 \in \{1, \ldots, R\}$. Each row has entries:
\begin{equation}
    \label{eq:A_rho_entry}
    {A_\rho^{(s_1)}}_{(\ell_2, \ell_3, s_2, s_3), m_1} = \sum_{\substack{m_2+m_3 = m_1 \\ |m_i| \le \ell_i}}  \, \langle \ell_2 \, m_2 \, \ell_3 \,m_3 | \ell \,m_1 \,  \rangle \, x^{\ell_2}_{m_2}[s_2] \, x^{\ell_3}_{m_3}[s_3],
\end{equation}
for $m_1 \in \{-\ell, \ldots, \ell\}$, where $x^{\ell_2}_{m_2}[s_2]$ and $x^{\ell_3}_{m_3}[s_3]$ denote known signal coefficients from lower bands. The right-hand side vector is the empirical second moment components for bands $(\ell, \ell_2, \ell_3)$ and shells $(s_1, s_2, s_3)$:
\begin{equation}
    \label{eq:b_rho}
    [b_\rho^{(s_1)}]_{(\ell_2, \ell_3, s_2, s_3)} = [\hat{M}_2]_{\ell, \ell_2, \ell_3}(s_1, s_2, s_3).
\end{equation}

\textbf{Signal recovery.} Conversely, if the distribution Fourier coefficients are known up to band $\ell$ and the signal coefficients are known up to band $\ell-1$, we recover the signal coefficients at band $\ell$ by solving $R$ independent linear systems, one for each shell $s \in \{1, \dots, R\}$. For each shell $s$, equation~\eqref{2ndmomequation} yields:
\begin{equation}
\label{eq:linear_x}
A_{x}^{(s)}(\hat{\rho}_{0 \dots \ell}, x_{0 \dots \ell-1}) \, x_{\ell}[s] = b_{x}^{(s)}(\hat{M}_2),
\end{equation}
where $x_{\ell}[s] = (x^\ell_{-\ell}[s], \ldots, x^\ell_{\ell}[s])^\top$ 
is the $s$-th column of the signal matrix $X_\ell$. The matrix $A_x^{(s)}$ has $(2\ell+1)$ columns indexed by $m \in \{-\ell, \ldots, \ell\}$, and rows indexed by sextuples $(\ell_1, \ell_2, \ell_3, s_1, s_2, s_3)$ satisfying the following constraints: $0 \le \ell_1 < \ell$; $|\ell_1 - \ell_2| \le \ell_3 \le \min(\ell_1 + \ell_2, L)$ (triangle inequality); exactly one of $\{\ell_2, \ell_3\}$ equals $\ell$ while the other is less than $\ell$; $s_1 \in \{-\ell_1, \ldots,\ell_1\}$ and $s_2, s_3 \in \{1, \ldots, R\}$; and crucially, the shell index corresponding to band $\ell$ equals $s$ (i.e., $s_2 = s$ if $\ell_2 = \ell$, or $s_3 = s$ if $\ell_3 = \ell$). Without loss of generality, assume $\ell_2 = \ell$ and $s_2 = s$. Then each row has entries:
\begin{equation}
\label{eq:A_x_entry}
{A_x^{(s)}}_{(\ell_1, \ell_2, \ell_3, s_1, s_2, s_3), m_2} = \sum_{\substack{m_1+m_3 = m_2 \\ |m_i| \le \ell_i}}  \, \langle \ell_1 \, m_1 \, \ell_3 \,m_3 | \ell \,m_2 \,  \rangle \, \overline{\hat{\rho}^{\ell_1}_{m_1}[s_1]} \, x^{\ell_3}_{m_3}[s_3],
\end{equation}
for $m_2 \in \{-\ell, \ldots, \ell\}$, where $\hat{\rho}^{\ell_1}_{m_1}[s_1]$ denote known distribution Fourier coefficients and $x^{\ell_3}_{m_3}[s_3]$ denote known signal coefficients from the lower band $\ell_3 < \ell$. The case where $\ell_3 = \ell$ and $s_3 = s$ (with $\ell_2 < \ell$) is constructed analogously. The right-hand side vector consists of the empirical second moment components:
\begin{equation}
    \label{eq:b_x}
    [b_x^{(s)}]_{(\ell_1, \ell_2, \ell_3, s_1, s_2, s_3)} = [\hat{M}_2]_{\ell_1, \ell_2, \ell_3}(s_1, s_2, s_3).
\end{equation}

A pseudo-code of the algorithm is described in Algorithm \ref{alg:frequency_marching_2nd}.

\begin{algorithm}
\caption{Joint recovery of volume and distribution from the second moment
}
\begin{algorithmic}[1]
\State \textbf{Inputs:} Noisy observations $\{y_i\}_{i=1}^n$, noise variance $\sigma^2$, number of shells~$R$, maximal frequency~$L$. 
\State \textbf{Outputs:} Signal $x$ and distribution $\rho$ up to $L$.
\Statex
\State Compute and debias the second moment according to \eqref{eq:debias}.
\State Initialize $\hat{\rho}_0 = 1$. 
\State Estimate~$x_0^0[s], 1 \le s \le R$ from the first moment according to~\eqref{eq:x_0}.
\State Estimate~$x_m^1[s], 1 \le s \le R, m \in \{-1, 0, 1\}$ by applying matrix decomposition to the Gram matrix defined in~\eqref{eq:gram_1}.
\For{$\ell = 1: L$}
    \For{$s_1 = 1 :2\ell+1$}
        \State Solve the system of equations~\eqref{eq:linear_rho}: $\widehat{\hat{\rho}}_{\ell}[s_1] \gets \arg \min_{\rho} \| A_{\rho}^{(s_1)} \hat{\rho} - b_{\rho}^{(s_1)} \|_2^2$, where~$A_{\rho}^{(s_1)}$ is given in~\eqref{eq:A_rho_entry} and~$b_\rho^{(s_1)}$ is given in~\eqref{eq:b_rho}.
    \EndFor
    \If{$\ell \ge 2$}
        \For{$s = 1 \to R$}
            \State Solve the system of equations~\eqref{eq:linear_x}: $\hat{x}_{\ell}[s] \gets \arg \min_{x} \| A_{x}^{(s)} x - b_{x}^{(s)} \|_2^2$, where~$A_{x}^{(s)}$ is given in~\eqref{eq:A_x_entry} and~$b_x^{(s)}$ is given in~\eqref{eq:b_x}.
        \EndFor
    \EndIf
\EndFor
\State \Return $x, \rho$
\end{algorithmic}
\label{alg:frequency_marching_2nd}
\end{algorithm}

\subsection{Experimental setup}
We generated synthetic data using two volumes: the Plasmodium falciparum 80S ribosome~\cite{wong2014cryo} available from the Electron Microscopy Data Bank (EMDB)\footnote{\url{https://www.ebi.ac.uk/emdb/}} under accession code~\mbox{EMD-2660}, and the TRPV1 structure~\cite{gao2016trpv1}, available as \mbox{EMD-8117}. The volumes were downsampled to $31^3$ voxels, and expanded into the truncated 3-D spherical-Bessel expansion of \eqref{eq:3-d-fourier-bessel} up to degree $L$ with $R$ radial shells;~$L$ and~$R$ vary between the experiments. 
All experiments were performed on a Linux server equipped with an Intel Xeon Gold 6252 CPU @ 2.10GHz and 1.5 TB of RAM. The algorithms were implemented in MATLAB. The code used to reproduce all numerical experiments is publicly available at

\begin{center}
\url{https://github.com/krshay/orbit-recovery-so3}.    
\end{center}

Our observations follow~\eqref{eq:samples}; the observations are generated in the coefficients domain. Unless stated otherwise, we simulate non-uniform viewing distributions by sampling the Euler angles~$(\alpha, \beta, \gamma)$ independently from a zero-mean Gaussian distribution with a standard deviation of~$1$. The signal-to-noise ratio (SNR) is defined by
\begin{equation}
    \text{SNR} = \frac{\| x\|_\text{F}}{\|\varepsilon\|_\text{F}},
\end{equation}
where~$x$ denotes the expansion coefficients,~$\varepsilon$ the additive noise in coefficient space, and~\mbox{$\|\cdot\|_{\text{F}}$} the Frobenius norm. The recovery error is computed as
\begin{equation}
\text{Relative Error} = \frac{\|x - \hat{x}\|_\text{F}}{\|x\|_\text{F}},
\end{equation}
where~$x$ denotes the ground-truth expansion coefficients and~$\hat{x}$ the recovered coefficients. We note that the global rotation ambiguity is resolved in our implementation by fixing the coefficients for $\ell=0$ and $\ell=1$ to their true values, rendering a search over $\SO(3)$ unnecessary.

\subsection{Recovery from a set of noisy measurements}
We begin by demonstrating a successful volume reconstruction from a set of $n=50{,}000$ noisy measurements with~$\text{SNR} = 1/2$; see Figures~\ref{fig:reconstruction_80S} and~\ref{fig:reconstruction_TRPV1}. These results use $L = 13$ and $R = 8$. The molecular visualizations were produced using UCSF Chimera~\cite{pettersen2004ucsf}. The code ran for approximately $16$ hours. The reconstruction error was~$8.1 \%$ for the Plasmodium falciparum 80S ribosome, and~$21\%$ for the TRPV1 volume.

\begin{figure}[htbp]
    \centering
    \begin{subfigure}{0.48\linewidth}
        \centering
        \includegraphics[width=\linewidth, keepaspectratio]{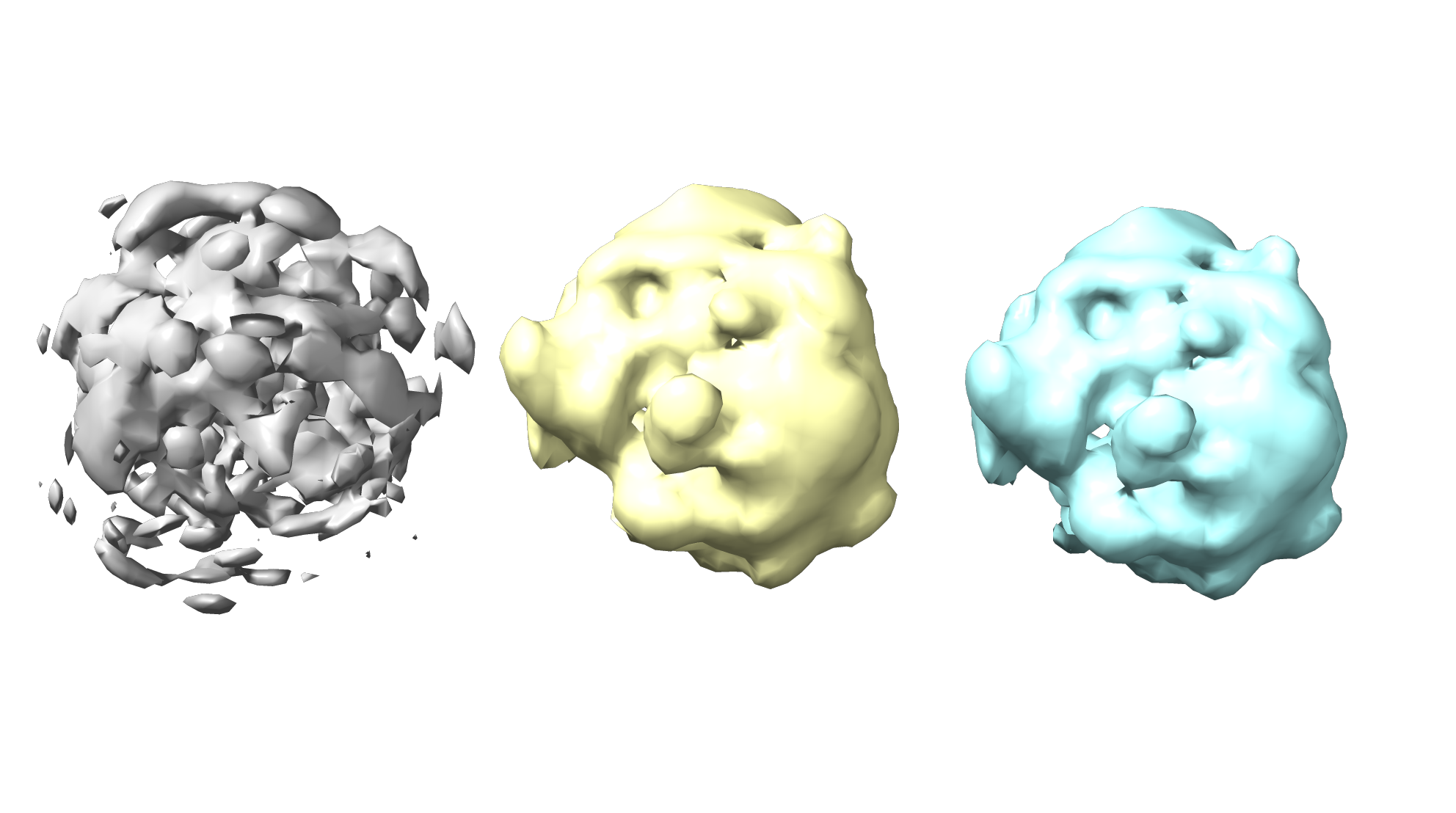}
        \caption{Plasmodium falciparum 80S ribosome}
        \label{fig:reconstruction_80S}
    \end{subfigure}
    \hfill 
    \begin{subfigure}{0.48\linewidth}
        \centering
        \includegraphics[width=\linewidth, keepaspectratio]{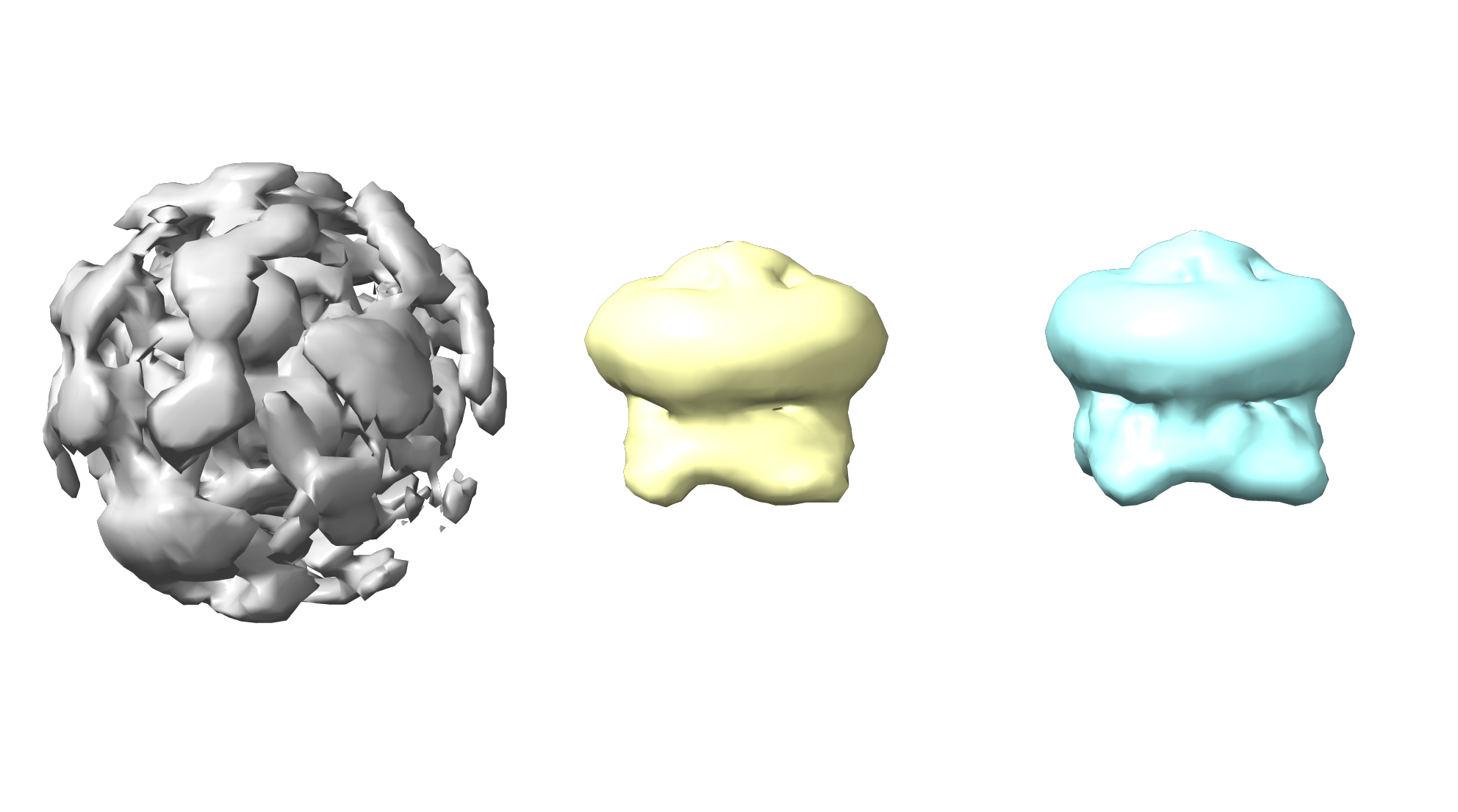}
        \caption{TRPV1 volume}
        \label{fig:reconstruction_TRPV1}
    \end{subfigure}
    
    \caption{Volume reconstruction results from $n=50{,}000$ noisy measurements with $\text{SNR} = 1/2$, using expansion parameters $L = 13$ and $R = 8$ shells. (a) Plasmodium falciparum 80S ribosome reconstruction. (b) TRPV1 volume reconstruction. In both panels, from left to right: a representative noisy measurement; the reconstructed volume; and the ground truth volume expanded to $L$.}
    \label{fig:reconstruction_combined}
\end{figure}

\begin{table}[htbp]
\centering
\caption{Condition numbers for signal and distribution bands for the reconstruction of the Plasmodium falciparum 80S ribosome.}
\label{tab:cond_comparison}
\begin{tabular}{c ccc ccc}
\toprule
 & \multicolumn{3}{c}{Signal Band Cond. Num.} & \multicolumn{3}{c}{Dist. Band Cond. Num.} \\
\cmidrule(lr){2-4} \cmidrule(lr){5-7}
$\ell$ & $R=3$ & $R=5$ & $R=8$ & $R=3$ & $R=5$ & $R=8$ \\
\midrule
1 & - & - & - & 5.8 & 5.9 & 5.4 \\
2 & 3.8 & 4.4 & 4.0 & 7.2 & 4.6 & 4.4 \\
3 & 3.8 & 3.5 & 3.2 & 8.7 & 5.7 & 4.7 \\
4 & 3.9 & 3.7 & 3.7 & 5.6 & 4.3 & 3.8 \\
5 & 4.3 & 4.1 & 4.2 & 7.7 & 6.5 & 5.7 \\
6 & 3.8 & 3.3 & 3.4 & 7.5 & 5.6 & 4.5 \\
7 & 4.4 & 4.0 & 3.7 & 8.5 & 6.9 & 6.2 \\
8 & 3.8 & 3.7 & 3.6 & 7.7 & 6.0 & 4.9 \\
9 & 4.1 & 4.1 & 4.0 & 7.9 & 5.8 & 4.9 \\
10 & 4.1 & 4.5 & 4.3 & 6.9 & 5.2 & 4.3 \\
11 & 3.5 & 3.8 & 4.3 & 5.9 & 4.0 & 3.8 \\
12 & 3.2 & 3.8 & 3.7 & 5.8 & 4.7 & 4.0 \\
13 & 3.2 & 3.5 & 3.7 & 6.5 & 4.9 & 4.4 \\
\bottomrule
\end{tabular}
\end{table}

\begin{table}[htbp]
\centering
\caption{Condition numbers for signal and distribution bands for the reconstruction of the TRPV1 volume.}
\label{tab:cond_comparison_trpv1}
\begin{tabular}{c ccc ccc}
\toprule
 & \multicolumn{3}{c}{Signal Band Cond. Num.} & \multicolumn{3}{c}{Dist. Band Cond. Num.} \\
\cmidrule(lr){2-4} \cmidrule(lr){5-7}
$\ell$ & $R=3$ & $R=5$ & $R=8$ & $R=3$ & $R=5$ & $R=8$ \\
\midrule
1 & - & - & - & 9.0 & 8.7 & 8.7 \\
2 & 5.5 & 6.6 & 6.6 & 31.5 & 23.2 & 21.2 \\
3 & 26.3 & 22.1 & 19.3 & 33.8 & 31.2 & 26.0 \\
4 & 17.4 & 12.3 & 11.8 & 40.9 & 28.5 & 21.7 \\
5 & 14.9 & 9.2 & 7.4 & 35.3 & 23.8 & 16.9 \\
6 & 11.3 & 9.4 & 9.1 & 27.3 & 27.1 & 19.7 \\
7 & 14.4 & 15.0 & 14.2 & 34.8 & 37.2 & 27.9 \\
8 & 11.3 & 12.5 & 13.4 & 22.2 & 19.6 & 16.0 \\
9 & 7.1 & 10.3 & 12.2 & 23.1 & 22.8 & 15.2 \\
10 & 3.9 & 5.6 & 6.7 & 27.3 & 20.2 & 14.7 \\
11 & 3.6 & 4.6 & 4.8 & 25.9 & 25.6 & 18.5 \\
12 & 3.3 & 4.0 & 4.3 & 26.3 & 20.6 & 17.6 \\
13 & 4.9 & 3.3 & 3.6 & 32.9 & 21.8 & 22.7 \\
\bottomrule
\end{tabular}
\end{table}

\subsection{Condition number analysis}
The proposed algorithm is based on successively solving systems of linear equations. Thus, the stability of the algorithm is based on the condition number of the linear systems. 
The key observation is that as $R$ increases, the number of rows in the distribution matrix grows as $O(R^2)$, while the number of columns remains fixed at $2L+1$. Likewise, the number of columns in the signal  matrix grows as $O(R)$. 
As a result, the condition number should decrease with more shells. This can be illustrated by the following simple result.



\begin{proposition}\label{prop:condition_heuristic}
Let $A \in \mathbb{R}^{m \times n}$ be a matrix with independent, identically distributed entries having mean zero and variance $\sigma^2$. Then the condition number $\kappa(A)$ satisfies
\[
\kappa(A) \to 1 \quad \text{almost surely as } m \to \infty \text{ with } n \text{ fixed}.
\]
\end{proposition}

\begin{proof}
By the law of large numbers, the Gram matrix satisfies $\frac{1}{m} A^T A \to \sigma^2 I_n$ almost surely implying $\kappa(A) \to 1$.
\end{proof}

Naturally, the entries of $M$~\eqref{eq:coeff_matrix} are not genuinely i.i.d.; rather, they are Clebsch-Gordan-weighted quadratic polynomials in the distribution and signal coefficients. Consequently, we cannot assert that the condition number approaches one as $R$ increases. Fortunately, our numerical experiments indicate that these systems are typically well conditioned. Indeed, since the equations remained well conditioned even when using the minimal number of radial shells, the dependence on $R$ appears to be relatively mild.

Tables~\ref{tab:cond_comparison} and~\ref{tab:cond_comparison_trpv1} report the condition numbers for the signal and distribution bands for the Plasmodium falciparum 80S ribosome and the TRPV1 volume, respectively. We evaluate these for varying radial resolutions $R \in \{3, 5, 8\}$, while fixing the maximum spherical harmonic degree at~$L=13$.
For the Plasmodium falciparum 80S ribosome, the condition numbers are low across all bands and shell counts. In the case of TRPV1, we observe higher condition numbers, particularly in the distribution band where values fluctuate between 20 and 40. This degradation in conditioning indicates that the inverse problem is more sensitive to noise for this structure, which is consistent with the higher reconstruction error observed for TRPV1 ($21 \%$) compared to the 80S ribosome ($8.1 \%$). However, despite the increased error, the condition numbers remain well below the threshold of numerical instability, allowing the algorithm to successfully recover the core structural features of the volume.

\subsection{Robustness to noise}
\begin{figure}[htbp]
\centering
	\includegraphics[width=0.7\columnwidth, keepaspectratio]{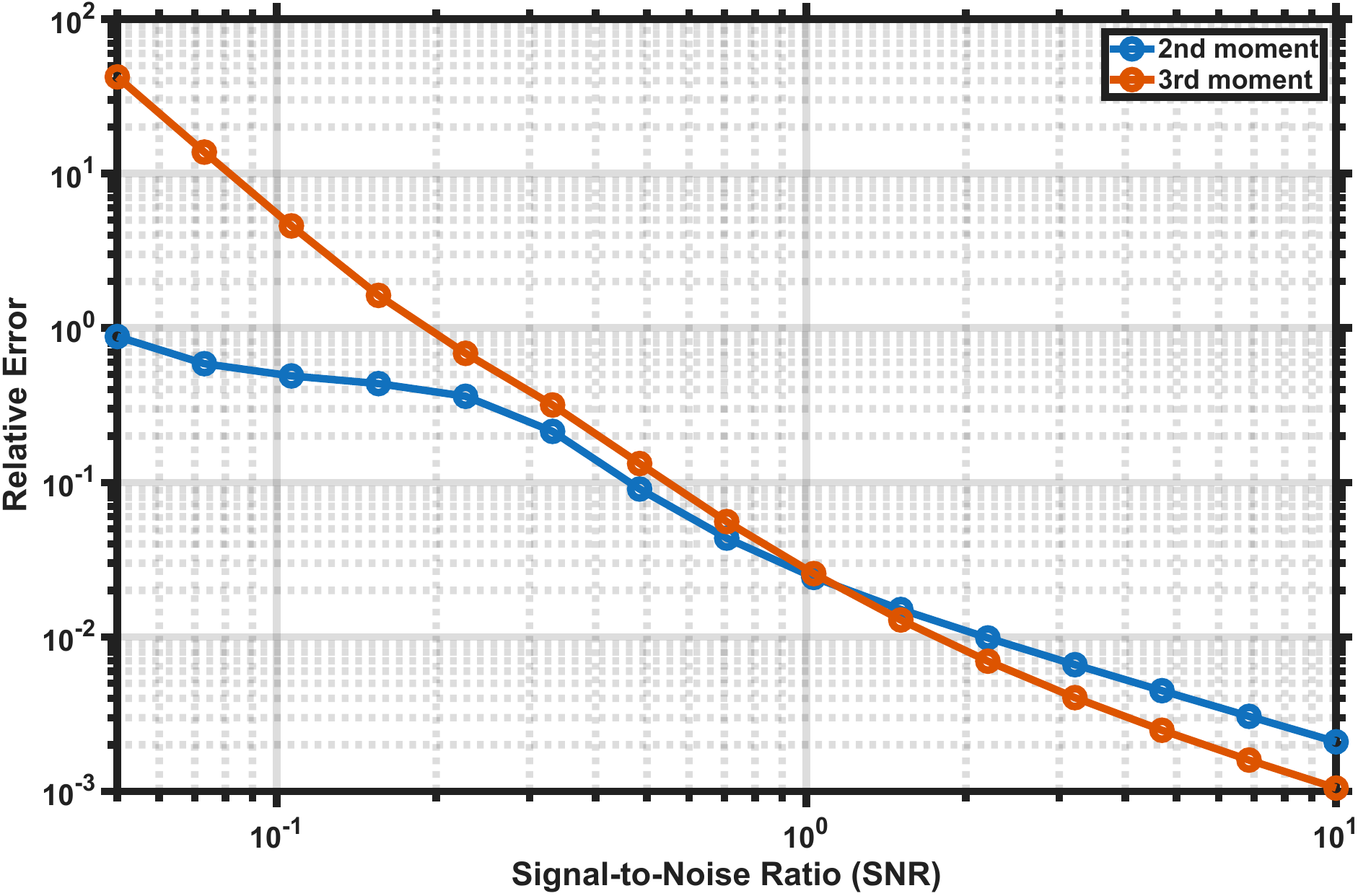}
	\caption{ 
    The average reconstruction error as a function of the SNR. The second-moment estimator is more robust in the high-noise regime ($\text{SNR} < 1$), while the third-moment estimator~\cite{bendory2025orbit} converges faster as the noise decreases.}
	\label{fig:error_vs_snr}
\end{figure}
We investigate the robustness of our algorithm against additive Gaussian noise, added to the observations. Figure~\ref{fig:error_vs_snr} illustrates the relative reconstruction error for different SNR levels. We compared the proposed algorithm (based on the  second moment) with the  algorithm  of~\cite{bendory2025orbit}; the latter algorithm uses the 
third moment and does not harness the non-uniform distribution of the rotations.  
For these experiments, the number of observations was fixed at $n=100{,}000$, and we used the Plasmodium falciparum 80S ribosome, expanded with $L=5$ and $R=5$. The reconstruction error was averaged over $5$ trials.

The results reveal two distinct regimes. In the low SNR regime ($\text{SNR} < 1$), the second-moment estimator significantly outperforms the third-moment method. This advantage is attributed to the fact that the third moment involves cubic terms of the data, which amplifies the noise more severely than the quadratic terms used in the second moment. Conversely, in the high SNR regime ($\text{SNR} > 1$), the third-moment estimator exhibits a steeper convergence rate and achieves a slightly lower reconstruction error. 

\subsection{Sample complexity}
We examine the reconstruction accuracy as a function of the sample size $n$. Figure~\ref{fig:error_vs_samples} displays the reconstruction error for the second-moment and third-moment methods as $n$ varies, with the SNR fixed at $1/2$. We used the Plasmodium falciparum 80S ribosome, expanded with~$L = 5$ and~$R = 5$. The reconstruction error was averaged over~$5$ trials.

\begin{figure}[htbp]
\centering
	\includegraphics[width=0.7\columnwidth, keepaspectratio]{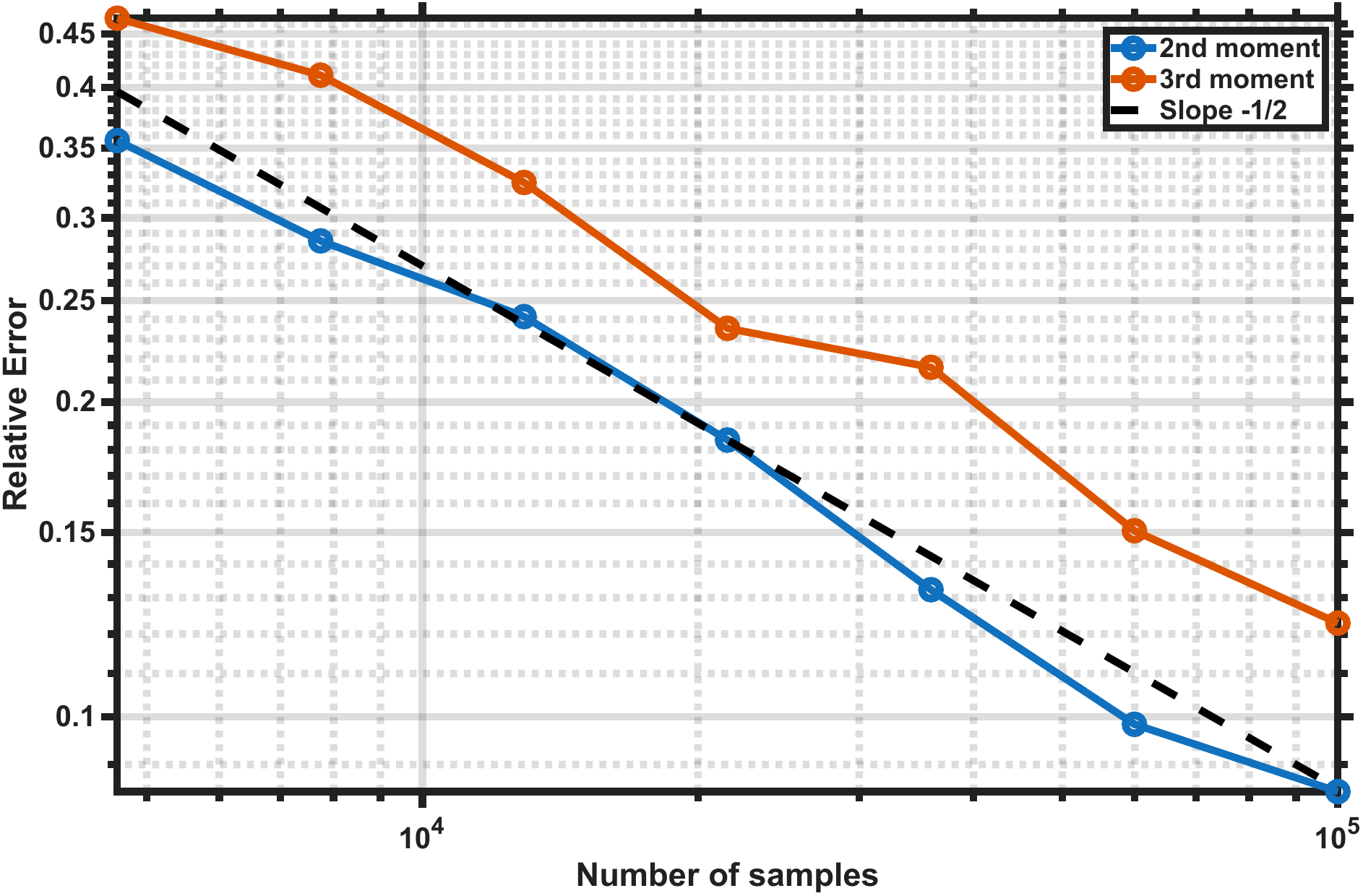}
\caption{The average reconstruction error as a function of the number of observations. The error decays at a rate of approximately $O(1/\sqrt{n})$ for both methods, with the second-moment estimator showing superior performance compared to the algorithm of~\cite{bendory2025orbit}.}
	\label{fig:error_vs_samples}
\end{figure}

Both estimators exhibit a convergence rate of approximately~$O(1/\sqrt{n})$. Notably, in this non-uniform experimental setting, the second-moment estimator consistently achieves lower reconstruction error than the third-moment method for a fixed sample size. This suggests that the second-moment method is more sample-efficient, at least in low-SNR environments, requiring fewer observations to achieve a target accuracy.

\subsection{Robustness to distribution nonuniformity}

Next, we evaluate the performance of the proposed algorithm for different distributions over $\SO(3)$. In our algorithm, we must estimate the Fourier matrices of the distribution from samples. 
Let $\{g_i\}_{i=1}^n$ be i.i.d. samples from the distribution with density $\rho\in L^2(G)$. Let $\varphi_V:G\rightarrow \Orth(V)$ be the representation of G on V. We estimate the Fourier matrix $\hat{\rho}(V)$ by
\begin{equation}
\hat{\rho}_n(V)=\frac{1}{n}\sum_{i=1}^n \varphi_V(g_i)^*.
\end{equation}
By the law of large numbers, as $n\to\infty$, we have $\hat{\rho}_n(V)\rightarrow \hat{\rho}(V)$.
 The variance of the estimator is determined by the following result:
 \begin{proposition} We have
 \begin{equation} 
 \mathrm{Var}(\hat{\rho}_n(V))=\frac{1}{n}(\dim(V)-||\hat{\rho}(V)||_F^2).   
 \end{equation}
 \end{proposition}
\begin{proof}
    Note that 
    \begin{equation}
        \begin{split}
       \mathrm{Var}(\hat{\rho}_n(V)) =  \mathrm{Var} \left(\frac{1}{n}\sum_{i=1}^n \varphi_V(g_i)^*\right)   =\frac{1}{n} \mathrm{Var} \varphi_V(g_i)^*).
        \end{split}
    \end{equation}
    The proof is completed by:
    \begin{equation}
    \begin{split}
     \mathrm{Var} \varphi_V(g_i)&=\mathbb{E}\Tr\varphi_V(g_i)\varphi_V(g_i)^* - \Tr\hat{\varphi}(V)\hat{\varphi}(V)^* \\&= \dim (V)-\|\hat{\rho}(V)\|_F^2.   
    \end{split}
    \end{equation}
        \end{proof}





The key takeaway is that distributions closer to uniform, where 
$||\hat{\rho}(H_{\ell})||_F^2$ is small for $\ell\geq 1$, exhibit higher variance in the Fourier coefficient estimates, whereas more concentrated distributions yield more accurate estimates. This explains the empirical observation that our second-moment algorithm performs better for concentrated distributions. Since distribution variance is therefore a key driver of second-moment estimation accuracy, we next present two sets of experiments to evaluate the robustness of the algorithm with respect to angular coverage and distribution variance.

\subsubsection{Robustness to distribution's angular coverage}

\begin{figure}[htbp]
\centering
	\includegraphics[width=0.7\columnwidth, keepaspectratio]{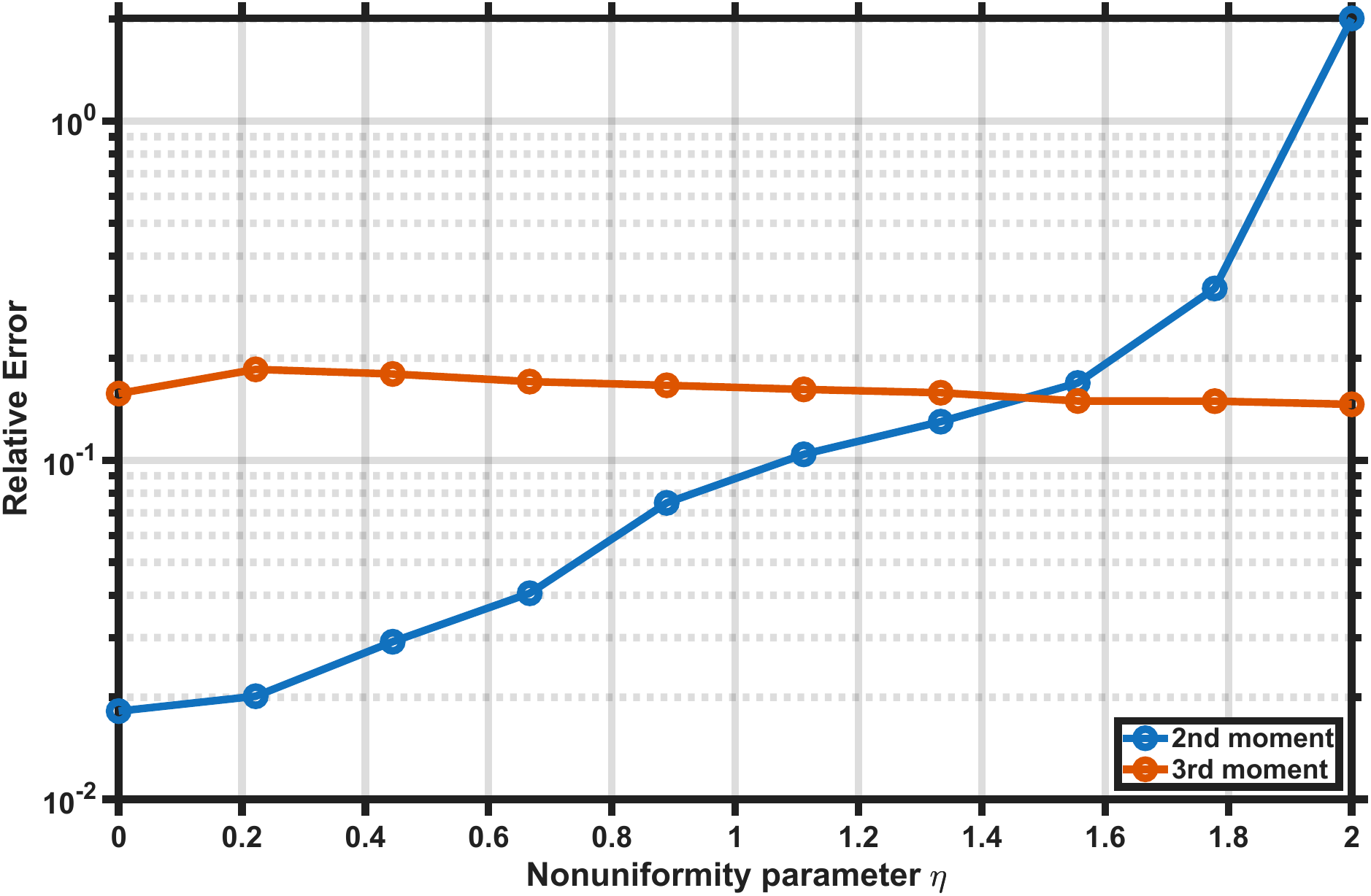}
	\caption{The average reconstruction error as a function of the nonuniformity parameter $\eta$, for the second-moment and third-moment methods. Small values of~$\eta$ (high non-uniformity) yield low errors for the second-moment method. As~$\eta \to 2$ (approaching uniformity), the second-moment method's error increases, whereas the third-moment method remains stable.
    }
	\label{fig:error_vs_coverage}
\end{figure}
We evaluate the reconstruction accuracy as a function of angular coverage by varying a parameter $\eta \in [0, 2]$ that scales the azimuthal range. Specifically, Euler angles are sampled as~$\alpha, \gamma \sim \text{Unif}[0, \eta\pi)$, while the polar angle is sampled to preserve the standard measure on the sphere:~$\beta = \cos^{-1}(1 - 2u)$ with~$u \sim \text{Unif}[0, 1)$. Thus, small $\eta$ corresponds to highly restricted views, while $\eta=2$ approaches a uniform distribution over $\SO(3)$. Figure~\ref{fig:error_vs_coverage} compares the second-moment method against the third-moment method using the Plasmodium falciparum 80S ribosome, with a volume expansion of~$L = 5$ and~$R = 5$, with~$n = 50{,}000$ observations and~$\text{SNR} = 1/2$. The reconstruction error was averaged over~$5$ trials. The second-moment estimator performs exceptionally well for restricted distributions~($\eta < 1$). However, as expected, its performance degrades as the distribution approaches uniformity~($\eta \to 2$). In contrast, the third-moment method (that does not harness the non-uniformity of the distribution) remains stable across all regimes, confirming that while higher moments are necessary for uniform distributions, the second moment is superior for retrieving signals from non-uniform rotation angles.

\subsubsection{Robustness to distribution's variance}
\begin{figure}[htbp]
\centering
	\includegraphics[width=0.7\columnwidth, keepaspectratio]{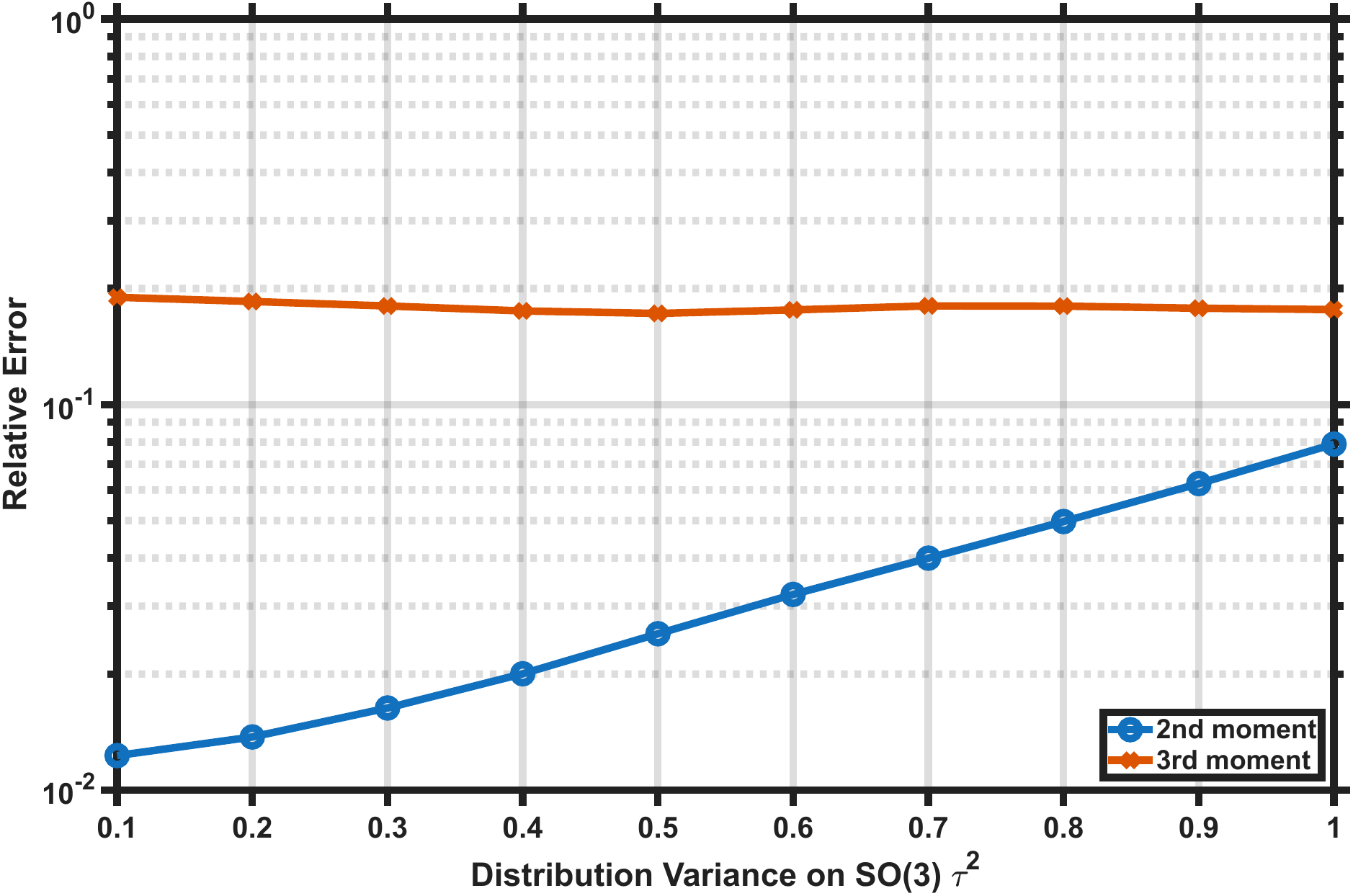}
	\caption{
    The average reconstruction error as a function of the non-uniformity parameter $\tau$ for the second-moment and third-moment methods. Small variance values of (representing high nonuniformity) yield low errors for the second-moment method. As the variance increases, the second-moment method's error increases, while the third-moment method remains mostly unaffected.}
	\label{fig:error_vs_variance}
\end{figure}
Next, we evaluate reconstruction accuracy as a function of the distribution variance on $\SO(3)$. We generated non-uniform viewing distributions by sampling Euler angles $(\alpha, \beta, \gamma)$ independently from a zero-mean Gaussian distribution with variance $\tau^2$, where $\tau^2$ varies. As $\tau$ decreases, the distribution becomes increasingly concentrated.

Figure~\ref{fig:error_vs_variance} compares the second-moment method against the third-moment method using the Plasmodium falciparum 80S ribosome, with a volume expansion of $L = 5$ and $R = 5$, with $n = 50,000$ observations and $\text{SNR} = 1/2$. The reconstruction error was averaged over 5 trials. As expected, the third-moment method remains stable across all distribution variances. The second-moment method outperforms the third-moment method for all distribution variances, although its performance degrades as the distribution becomes less concentrated (i.e., less nonuniform).

\section{The path to provable algorithms for cryo-EM?} \label{sec:cryo}
A major open question is whether similar techniques can recover a 3-D function (representing a molecular structure) from the second moment in the cryo-EM model. Unlike the setting considered here, cryo-EM involves a tomographic projection operator, which is not a group action. In particular, cryo-EM observations take the form
\begin{equation} \label{eq:cryo}
	y_i = P(g_i\cdot x) + \varepsilon_i, \quad i=1,\ldots,n,
\end{equation}
where $P$ denotes tomographic projection along the $z_3$-axis, i.e.,
$P(x(z_1,z_2))=\int_{z_3}x(z_1,z_2,z_3)dz_3$.
The distribution over rotations in cryo-EM is typically non-uniform~\cite{tan2017addressing,lyumkis2019challenges,barchet2026explicit}. Explicit expressions for the first and second moments of the cryo-EM model can be found in~\cite{sharon2020method,hoskins2024subspace}.

It is unclear whether the approach developed in this work---reducing moment inversion to a sequence of linear systems---extends to the cryo-EM setting. Nevertheless, alternative strategies may be possible. Below, we outline a prospective method that decomposes this  problem into an alternating sequence of linear subproblems.

The first moment of the cryo-EM model can be expressed as a convolution of the underlying structure with the distribution of rotations. If this distribution is known and sufficiently generic, then recovering the structure reduces to a linear inverse problem. While we are not aware of an established approach for directly estimating the rotation distribution, an initial estimate could be chosen at random, transferred from a related experiment, or learned by leveraging the large corpus of available reconstructed structures~\cite{iudin2023empiar,wwpdb2024emdb}.
Next, observe that the second moment depends linearly on the rotation distribution. Thus, given an estimate of the 3-D structure, the distribution can be refined by solving a linear system. This suggests an alternating scheme that iteratively updates the rotation distribution given the current structure estimate, and then updates the structure given the current distribution estimate. Determining, both numerically and analytically, when such a strategy succeeds remains an important open problem.

\section*{Acknowledgment}
T.B. and D.E. are supported in part by BSF under Grant 2020159. T.B. and N.S are supported in part by NSF-BSF under Grant 2024791. T.B is also supported in part by ISF under Grant 1924/21, and in part by a grant from The Center for AI and Data Science at Tel Aviv University (TAD). D.E was also partially supported by NSF grant DMS2205626. S.K. is supported by TAD Excellence Program for Doctoral Students in Artificial Intelligence and Data Science. N.S. is partially supported by the BSF award 2024266 and the DFG award 514588180.

\bibliographystyle{plain}

\end{document}